\documentclass[11pt]{article}

\usepackage[utf8]{inputenc}
\usepackage[T1]{fontenc}

\usepackage{fullpage}

\usepackage[small]{caption}
\usepackage{subcaption}

\usepackage{amsmath}
\usepackage{amsfonts}
\usepackage{amssymb}
\usepackage{amsthm}
\usepackage{amstext}
\usepackage{thmtools}

\usepackage{csquotes}
\usepackage{xcolor}
\usepackage{graphicx}
\usepackage{booktabs}

\usepackage{mathtools}
\usepackage{xspace}
\usepackage{enumerate}
\usepackage{algorithm}
\usepackage{algorithmic}
\usepackage{paralist}
\usepackage{bm}

\usepackage{url}
\usepackage{hyperref}
\usepackage{libertine}
\usepackage[libertine]{newtxmath}
\usepackage{microtype}

\usepackage[capitalize,noabbrev]{cleveref}

\theoremstyle{plain}
\newtheorem{theorem}{Theorem}[section]
\newtheorem{proposition}[theorem]{Proposition}
\newtheorem{lemma}[theorem]{Lemma}
\newtheorem{corollary}[theorem]{Corollary}
\newtheorem{observation}[theorem]{Observation}
\theoremstyle{definition}
\newtheorem{definition}[theorem]{Definition}

\theoremstyle{remark}

\usepackage[textsize=tiny,textwidth=1.6cm]{todonotes}

\usepackage{bm}
\usepackage{enumerate}
\usepackage{thm-restate}
\usepackage{xspace}
\usepackage{csquotes}

\DeclarePairedDelimiter{\abs}{\lvert}{\rvert}

\newcommand{\bigO}{\mathcal{O}} 
\newcommand{\A}{\ensuremath{\mathcal{A}}\xspace}

\newcommand{\B}{\ensuremath{\mathcal{B}}\xspace}

\newcommand{\alg}{{\normalfont\textsc{Alg}}}
\newcommand{\opt}{{\normalfont\textsc{Opt}}}

\newcommand{\hsigma}{\hat{\sigma}}

\newcommand{\hc}{{\bm\hat{c}}}

\newcommand{\hw}{\bm\hat{w}}
\newcommand{\hW}{\bm\hat{W}}
\newcommand{\hchains}{\bm\hat{\mathcal{C}}}

\newcommand{\chains}{\mathcal{C}}
\newcommand{\optspeed}{\alpha}
\newcommand{\numc}{\omega}
\newcommand{\hpreceq}{\bm\hat{\preceq}}
\newcommand{\dualSa}{\bm\bar{a}}
\newcommand{\dualSb}{\bm\bar{b}}
\newcommand{\dualSc}{\bm\bar{c}}
\newcommand{\dualVa}{a}
\newcommand{\dualVb}{b}
\newcommand{\dualVc}{c}
\newcommand{\df}{dual-fitting\xspace}
\newcommand{\frontjobs}{F}
\newcommand{\average}{a}
\newcommand{\averagep}{\ensuremath{\hat{a}}}

\newcommand{\largestinv}{\mathcal{L}}

\theoremstyle{plain}

\title{Minimalistic Predictions to Schedule Jobs with \\ Online Precedence Constraints}

\author{Alexandra Lassota~\thanks{EPFL, Switzerland, funded by the Swiss National Science Foundation project \emph{Complexity of Integer Programming}~(207365)} \and Alexander Lindermayr~\thanks{University of Bremen, Faculty of Mathematics and Computer Science, Germany. \{linderal,nmegow,jschloet\}@uni-bremen.de} \and Nicole Megow~\footnotemark[2] \and Jens Schl{\"o}ter~\footnotemark[2]}

\begin{document}

\maketitle

\thispagestyle{empty}

\begin{abstract}
We consider non-clairvoyant scheduling with \emph{online} precedence constraints, where an algorithm is oblivious to any job dependencies and learns about a job only if all of its predecessors have been completed. Given strong impossibility results in classical competitive analysis, we investigate the problem in a learning-augmented setting, where an algorithm has access to predictions without any quality guarantee.  
We discuss different prediction models: novel problem-specific models as well as general ones, which have been proposed in previous works. We present lower bounds and algorithmic upper bounds for different precedence topologies, and thereby give a structured overview on which and how additional (possibly erroneous) information helps for designing better algorithms. Along the way, we also improve bounds on traditional competitive ratios for existing algorithms.
\end{abstract}

\section{Introduction}

Cloud computing is a popular approach to outsource heavy computations to specialized providers~\cite{Hayes08}. Concepts like Function-as-a-Service (FaaS) offer users on demand the execution of complex computations in a specific domain~\cite{LynnRLE17,ShahradBW19}. Such tasks are often decomposed into smaller jobs, which then depend on each other by passing intermediate results. The structure of such tasks heavily relies on the users input and internal dependencies within the users system. It might require diverse jobs to solve different problems with distinct inputs.

From the providers perspective, the goal is thus to schedule jobs with different priorities and interdependencies which become known only when certain jobs are completed and their results can be evaluated. 
From a more abstract perspective, we face \emph{online precedence constraint scheduling}: new jobs arrive only if certain jobs have been completed but the set of jobs and their dependencies are unknown to the scheduler. As tasks might have different priorities, it is a natural objective to minimize the total (average) weighted completion time of the jobs.
We focus on \emph{non-clairvoyant} schedulers that do not know a job's processing requirement in advance~\cite{motwani1994nonclairvoyant}, and we allow \emph{preemptive} schedules, i.e., jobs can be interrupted and resumed later.
We present and analyze 
(non-)clairvoyant 
algorithms and prove impossibility results for this problem.

Competitive analysis is a widely used technique to assess the performance of online algorithms~\cite{Borodin98}. The \emph{competitive ratio} of an algorithm is the maximum ratio over all instances between its objective value and the objective value of an \emph{offline} optimal solution. In our setting, an offline optimal solution is the best schedule that can be computed with complete information and unbounded running time on the instance. We say that an algorithm is~$\rho$-competitive if its competitive ratio is at most~$\rho$.

It is not hard to see that for our problem, we cannot hope for good worst-case guarantees: consider an instance of~$n-1$ initially visible jobs with zero weight such that exactly one of these jobs triggers at its completion the arrival of a job with positive weight.
Since the initial jobs are indistinguishable, in the worst-case, any algorithm completes the positive-weight job last.
An offline optimal solution can distinguish the initially visible jobs and immediately processes the one which triggers the positive-weight job. 
This 
already
shows that no deterministic algorithm can have a better competitive ratio than~$\Omega(n)$ for~$n$ jobs. Notice that this strong impossibility result holds even for 
(seemingly) 
simple precedence graphs that consist of a collection of chains. In practice, such topology is highly relevant as, e.g., a sequential computer program executes a path (chain) of instructions that upon execution depends on the evaluation of control flow structures~(cf.~\cite{Allen1970control}).

To overcome such daunting lower bounds, we consider closer-to-real-world approaches to go beyond worst-case analysis. In particular, we study augmenting \emph{algorithms with predictions}~\cite{MitzenmacherV20,MitzenmacherV22}. The intuition is that in many applications, we can learn certain aspects of the uncertainty by considering historical data such as dependencies between jobs for certain computations and inputs. While these predictions might not reflect the current instance, they can contain enough information to design algorithms that break pessimistic worst-case lower bounds. Besides specifying the type of information, this requires a measure for a prediction's quality. This allows parameterized performance guarantees of algorithms w.r.t. the amount of information a prediction contains. Important performance indicators are \emph{consistency}, that is the competitive ratio for best-possible predictions, and \emph{robustness}, that is an upper bound on the competitive ratio for any prediction.

Despite the immense research interest in learning augmented algorithms \cite{alps}, the particular choice of prediction models remains often undiscussed. In this work, we discuss various models and analyze their strengths and weaknesses.
In particular, we present the first learning-augmented algorithms for scheduling with (online) precedence constraints.
The question driving our research~is:
\begin{quote}
	\emph{Which particular minimal information is required to achieve reasonable performance guarantees for scheduling with online precedence constraints?}
\end{quote}
Our starting point is the analysis of  
the two most common models,  \emph{full input} predictions, c.f.~\cite{PurohitSK18,AzarLT21,AzarLT22,Im0QP21,AntoniadisGS22,Bernardini22universal,ErlebachLMS22} and {\em action predictions}, c.f.~\cite{AntoniadisCE0S20,BamasMS20,LindermayrMS22,LindermayrM22,EberleLMNS22,Anand0KP22,JinM22}.
Our main focus is on a hierarchy of refined prediction models based on their entropy. That is, one can compute a prediction for a weaker model using a prediction from a stronger one, but not vice versa. We predict quantities related to the weight of unknown jobs which is in contrast to previous work which assumes predictions on the jobs' processing times  or machine speeds (except~\cite{ChoHS2022}).

For each prediction model, we analyze its power and limits by 
providing efficient algorithms and lower bounds on the best-possible performance guarantees w.r.t.~these models and the topological properties of the 
precedence constraints.

\subsection{Problem Definition and Prediction Models}

An instance of our problem is composed of a set~$J$ of~$n$ jobs and a precedence graph~$G=(J, E)$, which is an acyclic directed graph (DAG). Every job~$j \in J$ has a processing requirement~$p_j \geq 0$ and a weight~$w_j \geq 0$. An edge~$(j', j) \in E$ indicates that~$j$ can only be started if~$j'$ has been completed. If there is a directed path from~$j'$ to~$j$ in~$G$, then we say that~$j$ is a \emph{successor} of~$j'$ and that~$j'$ is a \emph{predecessor} of~$j$. If that path consists of a single edge, we call~$j$ and~$j'$ a \emph{direct} successor and predecessor, respectively.
For a fixed precedence graph~$G$, we denote by~$\numc$ the \emph{width} of~$G$, which is the length of the longest anti-chain in~$G$.

An algorithm can process a job~$j$ at a time~$t \geq r_j$ with a rate~$R_j^t \geq 0$, which describes the amount of processing the job receives at time~$t$. The completion time~$C_j$ of a job~$j$ is the first time~$t$ which satisfies~$\sum_{t' = 0}^t R_j^{t'} \geq p_j$. On a single machine a total rate of $1$ can be processed at any time~$t$, thus we require~$\sum_{j \in J} R_j^t \leq 1$. At any time~$t$ in a schedule, let~$\frontjobs_t = \left\{j \; \mid \; C_j > t \text{ and } \forall j' \text{ s.t. } (j', j) \in E \colon C_{j'} < t  \right\}$ denote the set of unfinished jobs without unfinished predecessors in~$G$.
We refer to such jobs as \emph{front jobs}. 
In the online setting, a job is revealed to the algorithm once all predecessors have been completed. The algorithm is completely oblivious to $G$, and, in particular, it does not know whether a front job has successors. Thus, at any time~$t$ an algorithm only sees jobs~$j\in$~$\frontjobs_t$ with weights~$w_j$ but {\em not} their processing times~$p_j$. Note that the sets~$\frontjobs_t$ heavily depend on an algorithm's actions. At the start time~$t=0$, an algorithm sees 
$\frontjobs_0$, and until the completion of the last job, it does not know the total number of jobs.
An algorithm can at any time~$t$ only process front jobs, hence we further require that~$R_j^t = 0$ for all~$j \in J \setminus \frontjobs_t$.
The objective of our problem is to minimize~$\sum_{j \in J} w_j C_j$.
For a fixed instance, we denote the optimal objective value by~$\opt$ and for a fixed algorithm, we denote its objective value by~$\alg$.

We study different topologies of precedence graphs. In addition to general DAGs, we consider \emph{in-forests} resp.\ \emph{out-forests}, where every node has at most one outgoing resp.\ incoming edge. Further, we study \emph{chains}, which is a precedence graph that is an in-forest and an out-forest simultaneously. 
If an in- or out-forest has only one connected component, we refer to it as \emph{in-} and \emph{out-tree}, respectively.

Two of the most studied prediction models are: a {\em full input prediction}, which is a prediction on the set of jobs with processing times and weights, and the complete precedence graph, and an \emph{action prediction}, which is a prediction on a full priority order over all jobs predicted to be part of the instance (\emph{static}) or a prediction on which job to schedule next whenever a machine idles (\emph{adaptive}).

Both prediction models require a significant amount of information on the input or an optimal algorithm. This might be unrealistic or costly  to obtain and/or not necessary. We aim for minimalistic extra information and quantify its power.

The set of front jobs $\frontjobs_0$ does not give sufficient information for obtaining a competitive ratio better than $\Omega(n)$, as shown above. For a job~$v \in \frontjobs_0$, we define the set~$S(v)$ consisting of $v$ and its successors,  
and we let $w(S(v)): = \sum_{u \in S(v)} w_u$.
We consider various predictions on the set~$S(v)$: \begin{list}{}{}
	\item[\emph{Weight predictions:}] Predictions $\hW_{v}$ on the total weight $w(S(v))$ of each front job $v \in \frontjobs_0$. 
	\item[\emph{Weight order predictions:}] 
The \emph{weight order}~$\preceq_0$ over~$\frontjobs_0$ sorts the jobs~$v \in \frontjobs_0$ by non-increasing~$w(S(v))$, i.e.,~$v \preceq_0 u$ implies~$w(S(v)) \geq w(S(u))$. We assume access to a prediction~$\hpreceq_0$ on~$\preceq_0$.
	\item[\emph{Average predictions:}] Predictions~$\averagep_v$ on the average weight $\average(S(v)) = \frac{\sum_{u \in S(v)} w_u }{\sum_{u \in S(v)} p_u }$ of each front job $v \in \frontjobs_0$. 
\end{list} 
For each of these three models, we distinguish \emph{static} and \emph{adaptive} predictions. Static predictions refer to predictions only on the initial front jobs~$\frontjobs_0$, and adaptive predictions refer to a setting where we receive access to a new prediction whenever a job becomes visible. 

\subsection{Our Results}

Our results can be separated into two categories. First, we consider the problem of scheduling with online precedence constraints with access to additional \emph{reliable} information. In particular, we consider all the aforementioned prediction models and design upper and lower bounds for the online problem enhanced with access to the respective additional information. We classify the power of the different models when solving the problem on different topologies. 

For the second type of results, we drop the assumption that the additional information is accurate and turn our pure online results into learning-augmented algorithms. We define suitable error measures for the different prediction models to capture the accuracy of the predictions, and give more fine-grained competitive ratios depending on these measures. 
We also extend our algorithms to achieve robustness.

Next, we give an overview of our results for these categories. We state all results for the single machine setting but show in Appendix~\ref{app:multiple} that they extend to identical parallel machines. 

\paragraph{Reliable additional information} 
\Cref{table:summary} summarizes our results for the pure online setting enhanced with reliable additional information. 
Our main 
results are a~$4$-competitive algorithms for chains and out-forests with weight predictions, and a~$H_{\numc}$-competitive algorithm for out-forests with adaptive weight order predictions, where~$H_k$ is the~$k$th harmonic number. 
The results show that additional information significantly improves the (worst-case) ratio compared to the setting with no predictions.

\begin{table}[tb]
	\caption{Summary of bounds on the competitive ratio given reliable information. We denote by~$P$ the total processing time and by~$H_k$ the~$k$th harmonic number.}
	\label{table:summary}
	\begin{center}
		\begin{tabular}{lll}
			\toprule
			Prediction Model & Topology & Bound  \\
			\midrule
			Actions & DAG & $\Theta(1)$ \\
			Input & DAG & $\Theta(1)$ \\
			Adaptive weights & Out-Forests & $\Theta(1)$\\
			Adaptive weights & In-Trees & $\Omega(\sqrt{n})$\\
			Static weights & Out-Trees & $\Omega(n)$ \\
			Static weights & Chains & $\Theta(1)$ \\
			Adaptive weight order & Out-Forests & $\bigO(H_\numc)$ \\
			Static weight order & Chains & $\bigO(H^2_\numc \sqrt{P})$ \\
Adaptive averages & Chains & $\Omega(\sqrt{n})$\\
			No Prediction & Chains & $\Omega(n)$  \\
			\bottomrule
		\end{tabular}
	\end{center}
\end{table}

Our main non-clairvoyant algorithm, given correct weight predictions, has a competitive ratio of at most $4$ for online out-forest precedence constraints on a single machine. This improves even for offline precedence constraints upon previous best-known bounds of~$8$~\cite{Jager21} and~10~\cite{GargGKS19} for this problem, although these bounds also hold in more general settings. 
To achieve this small constant, we generalize the Weighted Round Robin algorithm (WRR)~\cite{motwani1994nonclairvoyant,KimC03a} for non-clairvoyant scheduling \emph{without} precedence constraints, which advances jobs proportional to their weight, to our setting.
We handle each out-tree as a super-job and update its remaining weight when a sub-job completes. 
If the out-tree is a chain, this can be done even if only \emph{static} weight predictions are given. Otherwise, when an out-tree gets divided into multiple remaining out-trees, the distribution of the remaining weight is unknown, thus we  have to rely on \emph{adaptive} predictions.
Due to the increased dynamics of gaining partial weight of these super-jobs, the original analysis of WRR is not applicable. Instead, we use the \df technique, which has been previously used for offline precedence constraints~\cite{GargGKS19}. While their analysis builds on offline information and is infeasible in our model, we prove necessary conditions on an algorithm to enable the \df, which are fulfilled even in our limited information setting. Surprisingly, we also show that a more compact linear programming (LP) relaxation, which does not consider transitive precedences, is sufficient for our result. In particular, compared to the LP used in~\cite{GargGKS19}, it allows us to craft simple duals which do not involve gradient-type values of the algorithm's rates.

In the more restricted model of weight order predictions, WRR cannot be applied, as the rate computation crucially relies on \emph{precise} weight values. We observe, however, that WRR's rates at the start of an instance have the same ordering as the known chain order. We show that guessing rates for chains in a way that respects the ordering compromises only a factor of at most~$H_\numc$ in the competitive ratio. 
If the weight order is adaptive, we show a competitive ratio of~$4 \cdot H_\numc$. Otherwise, we give a worse upper bound and evidence that this might be best-possible for this algorithm.

\paragraph{Learning-augmentation} We extend our algorithmic results by designing suitable error measures for the different prediction models and proving error-dependent competitive ratios.
Finally, we show how existing techniques can be used to give these algorithms a robustness of~$\mathcal{O}(\omega)$ at the loss of only a constant factor in the error-dependent guarantee. Note that a robustness~$\mathcal{O}(\omega)$ matches the lower bound for the online problem without access to additional information.

\subsection{Further Related Work}

Scheduling jobs with precedence constraints to minimize the sum of (weighted) completion times has been one of the most studied scheduling problems for more than thirty years. The offline problem is known to be NP-hard, even for a single machine~\cite{lawler78,lenstrR78}, and on two machines, even when precedence constraints form chains~\cite{duLY91, timkovsky03}. 
Several polynomial-time algorithms based on different linear programming formulations achieve an approximation ratio of~$2$ on a single machine, whereas special cases are even solvable optimally; we refer to~\cite{CorreaS05,AmbuhlM09} for comprehensive
overviews.  For scheduling on $m$ parallel identical machines, the best known approximation factor is $3-1/m$~\cite{hallSSW97}.

For scheduling with online precedence constraints, strong and immediate lower bounds rule out a competitive ratio better than $\Omega(n)$ for the min-sum objective. Therefore, online scheduling has been mainly studied in a setting where jobs arrive online but once a job arrives its processing time, weight, and relation to other (already arrived) jobs is revealed~\cite{hallSSW97,ChakrabartiPSSSW96,BienkowskiKL21}.
Similarly, we are not aware of any previous learning-augmented algorithm for online precedence constraints and/or weight predictions. Previous work on minimizing the total completion time focussed on the setting where jobs either arrive online and/or are clairvoyant~\cite{PurohitSK18,Im0QP21,LindermayrM22,DinitzILMV22portfolios,BampisDKLP22}.

\section{Robustness via Time-Sharing}\label{sec:time-sharing}

Before we move to concrete algorithmic results,
we quickly argue that any~$\rho$-competitive algorithm for scheduling with online precedence constraints of width $\numc$ can be extended to a~$\bigO(\min\{\rho,\omega\})$-competitive algorithm. In particular, if~$\rho$ depends on a prediction's quality, this ensures that this algorithm is robust against arbitrarily bad predictions.

To this end, consider the algorithm which at any time~$t$ shares the machine equally among all front jobs~$F_t$, i.e., gives every job~$j \in F_t$ rate~$R_j^t = \frac{1}{\abs{F_t}} \ge \frac{1}{\numc}$. For a fixed job~$j$, compared to its completion time in a fixed optimal schedule, the completion time in the algorithm's schedule can be delayed by at most a factor of~$\numc$. 
We conclude:

\begin{proposition}\label{lemma:chain-robust}
	There is an~$\numc$-competitive non-clairvoyant single-machine algorithm for minimizing the total weighted completion time of jobs with online precedence constraints.
\end{proposition}

We can now use the \emph{time-sharing technique} to combine this $\omega$-competitive algorithm with any other algorithm for scheduling online precedence constraints while retaining the better competitive ratio of both up to a factor of~$2$.

\begin{theorem}[\cite{PurohitSK18,LindermayrM22}]\label{thm:alg-combination}
Given two deterministic algorithms with competitive ratios~$\rho_\A$ and~$\rho_\B$ for minimizing the total weighted completion time with online precedence constraints on identical machines, there exists an algorithm for the same problem with a competitive ratio of at most~$2 \cdot \min\{\rho_\A, \rho_\B\}$.
\end{theorem}

In \cite{PurohitSK18,LindermayrM22} there is an additional monotonicity requirement which we claim to be unnecessary; see \Cref{app:monotonicity}.

\section{Action Predictions}

Action predictions give an optimal algorithm. Hence, following accurate predictions clearly results in an optimal solution.
To define an error measure for erroneous static and adaptive action predictions, let~$\hsigma: J \to [n]$ be the order in which a fixed 
static or adaptive action prediction suggests to process jobs. 
In case of static action predictions, we receive the predicted order initially, meaning that it might predict a set of jobs~$\hat{J}$ different to the actual~$J$. During the analysis, we can simply remove the jobs~$\hat{J}\setminus J$ from~$\hsigma$ as they do not have any effect on the schedule for the actual instance. For the jobs in~$J\setminus\hat{J}$, we define the static action prediction algorithm to just append them to the end of the order~$\hsigma$ once they are revealed. Thus, we can still treat~$\hsigma$ as a function from~$J$ to~$[n]$.
We analyse an algorithm which follows a static or adaptive 
action prediction using the permutation error introduced in \cite{LindermayrM22}.
To this end, let~$\sigma: J \to [n]$ be the order of a fixed optimal solution for instance~$J$, and~$\mathcal{I}(J, \hsigma) = \{ (j', j) \in J^2 \mid \sigma(j') < \sigma(j) \land \hsigma(j') > \hsigma(j) \}$ be the set of inversions between the permutations~$\sigma$ and $\hsigma$. 
Applying the analysis of~\cite{LindermayrM22} yields the following theorem.

\begin{theorem}
	\label{thm:action}
	Given static or adaptive action predictions, there exists an efficient $\bigO(\min \left\{ 1 + \eta, \numc \right\})$-competitive non-clairvoyant algorithm for minimizing the total weighted completion time on a single machine with online precedence constraints, where $\eta = \sum_{(j', j) \in \mathcal{I}(J, \hsigma)} (w_{j'} p_j - w_j p_{j'})$.
\end{theorem}

\section{Full Input Predictions}

We can use full input predictions to compute static action predictions $\hsigma$. In general, computing  $\hsigma$ requires exponential running time as the problem is NP-hard. For special cases, e.g., chains, there are efficient algorithms~\cite{lawler78}. 

While following $\hsigma$ allows us to achieve the guarantee of~\Cref{thm:action}, the error $\eta$ does not directly depend on the predicted input but on an algorithm which computed actions for that input.
Thus, we aim at designing error measures depending directly on the \enquote{similarity} between the predicted and actual instance.
As describing the similarity between two graphs is a notoriously difficult problem on its own, we leave open whether there is a meaningful error for general topologies. 
However, we give an error measure for chains.
The key idea of this error is to capture additional cost that any algorithm pays due to both, \emph{absent predicted} weights and \emph{unexpected actual} weights. This is in the same spirit as the universal cover error for graph problems in \cite{Bernardini22universal}.
Assuming that the predicted and actual instance only differ in the weights, our error $\Lambda = \Gamma_u + \Gamma_a$ considers the optimal objective values $\Gamma_u$ and $\Gamma_a$ for the problem instances that use $\{(w_j - \hw_j)_+\}_j$ and  $\{(\hw_j - w_j)_+\}_j$ as weights, respectively. Then, $\Gamma_u$ and $\Gamma_a$ measure the cost for \emph{unexpected} and \emph{absent} weights.
In the appendix, we generalize this idea to also capture other differences of the predicted and actual chains and prove the following theorem.

\begin{restatable}{theorem}{thmInputPredErrorDep}\label{thm:input-predictions}
	Given access to an input prediction, there exists an efficient algorithm for minimizing the total weighted completion time of unit-size jobs on a single machine with online chain precedence constraints with a competitive ratio of at most $\bigO(\min \left\{ 1 + \Lambda, \numc \right\})$, where $\Lambda = \Gamma_u + \Gamma_a$.    
\end{restatable}

\section{Weight Value Predictions}

We now switch to more problem-specific prediction models, starting with weight value predictions.
We first prove strong lower bounds 
for algorithms with access to static weight predictions on out-trees and adaptive predictions on in-trees. Then, we give $4$-competitive algorithms
for accurate static predictions on chains, and adaptive weight predictions on out-forest precedence constraints, and  finally extend these results to obtain robust algorithms with error dependency.

The lower bound for out-trees adds a dummy root $r$ to the pure online lower bound composed of $\Omega(n)$ zero weight jobs, where exactly one hides a valuable job. In the static prediction setting we thus only receive a prediction for $r$, which does not help any algorithm to improve.

\begin{observation}
	\label{obs:lb-trees-static}
Any algorithm which has only access to static weight predictions has a competitive ratio of at least $\Omega(n)$, even if the precedence constraint graph is an out-tree.
\end{observation} 

For in-trees and adaptive weight predictions, we prove the following lower bound.

\begin{restatable}{lem}{LBInTrees}
	\label{lb:in-trees}
Any algorithm which has only access to adaptive weight predictions has a competitive ratio of at least~$\Omega(\sqrt{n})$, even for in-tree precedence constraints. \end{restatable}

\begin{proof}
	Consider an in-tree instance with unit-size jobs and root~$r$ of weight~$0$. There are~$\sqrt{n}$ chains of length~$2$ with leaf weights~$0$ and inner weights~$1$ which are connected to~$r$. Further, there are~$n-2\sqrt{n}-1$ leaves with weight 0, which are connected to a node~$v$ with weight~$1$, which itself is a child of~$r$.
	Note that the weight prediction for all potential front jobs except $r$ is always $1$. Thus, even the adaptive predictions do not help, and we can assume that the algorithm first processes the children of~$v$, giving a total objective of at least~$\Omega((n-2\sqrt{n}-1)^2 + (n-2\sqrt{n}-1)\sqrt{n}) = \Omega(n \sqrt{n})$, while processing the other leaves first yields a value of at most~$\bigO((2\sqrt{n})^2 + (2\sqrt{n} + n-2\sqrt{n})) = \bigO(n)$.
\end{proof}

\subsection{Algorithms for Reliable Information}
\label{sec:chains-consistent}

We present algorithms assuming access to \emph{correct} static or adaptive weight predictions and prove their competitiveness on online chain and out-forest precedence constraints using a unified analysis framework.
This uses a \df argumentation inspired by an analysis of an algorithm for {\em known} precedence constraints~\cite{GargGKS19}. The framework only requires a condition on the rates at which an algorithm processes front jobs, hence it is independent of the considered prediction model. Let~$U_t$ refer to the set of unfinished jobs at time~$t$, i.e.,~$U_t = \bigcup_{v \in F_t} S(v)$. Denote by~$w(J')$ the total weight of jobs in a set~$J'$. We write~$W(t)$ for~$w(U_t)$.

\begin{theorem}\label{thm:rates-algo}
    If an algorithm for online out-forest precedence constraints satisfies at every time~$t$ and~$j \in F_t$ that~$w(S(j)) \leq \rho \cdot R_j^t \cdot W(t)$, where~$R_j^t$ is the processing rate of~$j$ at time~$t$, it is at most~$4\rho$-competitive for minimizing the total weighted completion time on a single machine.
\end{theorem}

We first present algorithms for weight predictions and derive results using \Cref{thm:rates-algo}, and finally prove the theorem.

\paragraph{Static Weight Values for Chains}
We give an algorithm
for correct static weight predictions.
As \Cref{obs:lb-trees-static} rules out well-performing algorithms for out-tree precedence constraints with static weight predictions, we focus on chains. Correct static weight predictions mean access to the total weight~$W_c$ of every chain~$c$ in the set of chains~$\chains$.

\begin{algorithm}[tb]
	\caption{Weighted Round Robin on Chains}\label{alg:chain-robin}
	\begin{algorithmic}[1]
	\REQUIRE Chains~$\chains$, initial total weight $W_c$ for each $c \in \chains$.
	\STATE $t \gets 0$ and $W_c(t) \gets W_c$ for every $c \in \chains$.
	\WHILE{$U_t \neq \emptyset$}
	\STATE Process the front job of every chain $c$ at rate $\frac{W_c(t)}{\sum_{c} W_c(t)}$
	\STATE $t \gets t + 1$
	\STATE If a job $j$ in chain $c$ finished, $W_c(t) \gets W_c(t) - w_j$\label{line:chain-robin-update}
	\ENDWHILE
	\STATE Schedule remaining jobs in an arbitrary order.
	\end{algorithmic}
	\end{algorithm}

Algorithm~\ref{alg:chain-robin}, essentially, executes a classical weighted round robin algorithm where the rate at which the front job of a chain~$c$ is executed at time~$t$ is proportional to the total weight of unfinished jobs in that chain,~$W_c(t)$.
As this definition is infeasible for unfinished chains with $W_c(t)=0$, we process these in an arbitrary order in the end. As they have no weight, this does not negatively affect the objective.

Despite initially only having access to the weights~$W_c(t)$ for~$t=0$, the algorithm can compute~$W_c(t)$ for any~$t > 0$ by subtracting the weight of finished jobs of~$c$ from the initial~$W_c$ (cf.~Line~$5$).
Thus,~$W_c(t) = w(S(j))$ holds for any time~$t$ and every~$j \in F_t$, where~$c$ is the corresponding chain of job~$j$.
Further,~$W(t) = \sum_c W_c(t)$.
We conclude that, for any~$t$ and~$j \in F_t$, it holds~$R_j^t=\frac{w(S(j))}{W(t)}$. Using \Cref{thm:rates-algo} with~$\rho = 1$, we derive the following result:

\begin{theorem}\label{theorem:chain-robin}
     Given correct weight predictions, \Cref{alg:chain-robin} is a non-clairvoyant $4$-competitive algorithm for minimizing the total weighted completion time of jobs with online chain precedence constraints on a single machine.
\end{theorem}

\paragraph{Adaptive Weight Values for Out-Forests}

\Cref{obs:lb-trees-static} states that static weight predictions are not sufficient to obtain $\mathcal{O}(1)$-competitive algorithms for out-forests. The reason is that we, in contrast to chains, cannot recompute~$\hW_{j}$ whenever a new front job~$j$ appears. 
For adaptive predictions, however, we do not need to recompute~$\hW_{j}$, as we simply receive a new prediction. Thus, we can process every front job $j \in F_t$ with rate $R_j^t = \frac{\hW_j}{\sum_{j' \in F_t} \hW_{j'}}$. For correct predictions, \Cref{thm:rates-algo} directly implies the following. 

\begin{theorem}\label{theorem:out-forests-rr}
	    Given correct adaptive 
		weight predictions, there exists a non-clairvoyant $4$-competitive algorithm for minimizing the total weighted completion time of jobs with online out-forest precedence constraints on a single machine.	
\end{theorem}

\paragraph{Full proof of \Cref{thm:rates-algo}}

Fix an algorithm satisfying the conditions of \Cref{thm:rates-algo}.
Let $\alg$ be the objective value of the algorithm's schedule for a fixed instance. We introduce a linear programming relaxation similar to the one in~\cite{GargGKS19} for our problem on a machine running at lower speed~$\frac{1}{\optspeed}$, for some~$\optspeed \geq 1$. Let~$\opt_\optspeed$ denote the optimal objective value for the problem with speed~$\frac{1}{\optspeed}$.
As the completion time of every job is linear in the machine speed, we have~$\opt_\optspeed \leq \optspeed \cdot \opt$.
The variable~$x_{j,t}$ denotes the fractional assignment of job~$j$ at time~$t$. The first constraint ensures that every job receives enough amount of processing to complete, the second constraint restricts the avaible rate per time to $\frac{1}{\alpha}$, and the final constraint asserts that no job can be completed before its predecessors.
\begin{alignat}{3}
    \text{min} \quad &\sum_{j,t} w_{j} \cdot t \cdot \frac{x_{j,t}}{p_j} \tag{$\text{LP}_\optspeed$}\label{lp} \\ 
    \text{s.t.}  \quad & \sum_{t} \frac{x_{j,t}}{p_{j}} \geq 1   &&\forall j \notag \\
    & \sum_{j} \optspeed \cdot x_{j,t} \leq 1   &&\forall t \notag \\
    & \sum_{s \leq t} \frac{x_{j,s}}{p_j} \geq \sum_{s \leq t} \frac{x_{j',s}}{p_{j'}} && \quad \forall t, \forall (j,j') \in E \notag \\
    & x_{j,t} \geq 0 &&\forall j,t \notag 
\end{alignat}

The dual of~\eqref{lp} can be written as follows. \begin{alignat}{3}
    \text{max} \quad &\sum_{j} \dualVa_j - \sum_{t} \dualVb_{t} \tag{$\text{DLP}_\optspeed$}\label{dual} \\ 
    \text{s.t.}  \quad &\sum_{s \geq t} \left( \sum_{j': (j,j') \in E} \dualVc_{s,j \to j'} - \sum_{j': (j',j) \in E} \dualVc_{s,j' \to j} \right) \notag \\
    &\leq \optspeed \cdot \dualVb_{t} \cdot p_j - \dualVa_j + w_j \cdot t  \; \; \quad \forall j,t \label{constr:dual} \\ 
    &\dualVa_j, \dualVb_{t}, \dualVc_{t,j \to j'}  \geq 0  \quad  \qquad \qquad \forall t, \forall (j,j') \in E \notag
\end{alignat}

Let~$\kappa > 1$ be a constant which we fix later. We define a variable assignment for~\eqref{dual} 
as follows: $\dualSa_j = \sum_{s \geq 0} \dualSa_{j,s}$ for every job~$j$, where~$\dualSa_{j,s} = w_j$ if~$s \leq C_j$ and $\dualSa_{j,s} = 0$ otherwise,~$\dualSb_{t} = \frac{1}{\kappa} \cdot W(t)$ for every time~$t$, and~$\dualSc_{t, j' \to j} = w(S(j))$ if~$j,j' \in U_t$, and $\dualSc_{t, j' \to j} = 0$ otherwise, for every time~$t$ and edge~$(j',j) \in E$.

We show in the following that the variables~$(\dualSa_j,\dualSb_{t},\dualSc_{t,j \to j'})$ define a feasible solution for \eqref{dual} and achieve an objective value close to~$\alg$. Weak duality then implies \Cref{thm:rates-algo}. First, consider the objective value.

\begin{lemma}\label{lemma:chain-robin-dual-objective}
$\sum_{j} \dualSa_j - \sum_{t} \dualSb_{t} = (1 - \frac{1}{\kappa}) \alg$.
\end{lemma}

\begin{proof}
Note that~$\dualSa_j = w_j C_j$, and thus~$\sum_j \dualSa_j = \alg$. Also, since the weight~$w_j$ of a job~$j$ is contained in~$W(t)$ if~$t \leq C_j$, we conclude~$\sum_t \dualSb_{t} = \frac{1}{\kappa}\alg$.
\end{proof}

Second, we show that the duals are feasible for \eqref{dual}.

\begin{lemma}\label{lemma:chain-robin-dual-feasible}
   Assigning~$\dualVa_j = \dualSa_j$,~$\dualVb_{t} = \dualSb_{t}$ and~$\dualVc_{t,j \to j'} = \dualSc_{t,j \to j'}$ is feasible for~\eqref{dual} if~$\alpha \geq \kappa \rho$ and $ \frac{w(S(j))}{W(t)} \leq \rho R_j^t$ for any time~$t$  for~$j \in F_t$ and the algorithm's rates $R_j^t$. 
\end{lemma}

\begin{proof}
  Since our defined variables are non-negative by definition, it suffices to show that this assignment satisfies~\eqref{constr:dual}.
  Fix a job~$j$ and a time~$t \geq 0$. 
  By observing that~$\dualSa_j - t \cdot w_j \leq \sum_{s \geq t} \dualSa_{j,s}$, it suffices to verify
  \begin{equation}
    \sum_{s \geq t} \left(\dualSa_{j,s} +  \sum_{(j, j') \in E} \dualSc_{s,j \to j'} - \sum_{(j', j) \in E} \dualSc_{s,j' \to j} \right) \leq \optspeed \dualSb_{t} p_j. \label{eq:dual-red1}
\end{equation}

To this end, we consider the terms of the left side for all times~$s \geq t$ separately. For any~$s$ with~$s > C_j$, the left side of \eqref{eq:dual-red1} is zero, because~$\dualSa_{j,s} = 0$ and~$j \notin U_s$. 

Otherwise, if~$s \leq C_j$, let~$t_j^*$ be the first point in time after~$t$ when~$j$ is available, and let~$s \in [0, t_j^*)$. 
Then,~$j \in U_s$, and since each vertex in an out-forest has at most one direct predecessor, there must be a unique job~$j_1 \in U_s$ with~$(j_1, j) \in E$. Thus,~$\dualSc_{s,j_1 \to j} = w(S(j))$ and~$\dualSc_{s,j \to j'} = w(S(j'))$ for all~$(j, j') \in E$.
Observe that in out-forests, we have~$S(j') \cap S(j'') = \emptyset$ for all~$j' \neq j''$ with~$(j,j'),(j,j'') \in E$. This implies~$\sum_{(j,j') \in E} \dualSc_{s,j \to j'} = w(S(j)) - w_j$ and~$\sum_{(j,j') \in E} \dualSc_{s,j \to j'} - \sum_{(j_1, j) \in E} \dualSc_{s,j_1 \to j} = - w_j$. Hence,
  \[
    \dualSa_{j,s} + \sum_{(j, j') \in E} \dualSc_{s,j \to j'} - \sum_{(j', j) \in E} \dualSc_{s,j' \to j} \leq w_j - w_j = 0. 
  \]

Therefore, proving \eqref{eq:dual-red1} reduces to proving
\begin{equation}
    \sum_{s = t^*_j}^{C_j} \left(w_j + \sum_{(j, j') \in E} \dualSc_{s,j \to j'} - \sum_{(j', j) \in E} \dualSc_{s,j' \to j} \right) \leq \optspeed \dualSb_{t} p_j. \label{eq:dual-red2}
\end{equation}

Now, let~$s \in [t_j^*, C_j)$. There cannot be an unfinished job preceding~$j$, thus~$\sum_{(j', j) \in E} \dualSc_{s,j' \to j} = 0$. Observe that if there is a job~$j' \in U_s$ with~$(j, j') \in E$, the fact that~$j \in U_s$ implies~$j' \in U_s$, and thus~$\dualSc_{s,j \to j'} = w(S(j'))$ by definition. Using again the fact that the sets~$S(j')$ are pairwise disjoint for all direct successors~$j'$ of~$j$, i.e., for all~$(j,j') \in E$, this yields~$\sum_{(j, j') \in E} \dualSc_{s,j \to j'} = w(S(j)) - w_j$, and further gives
\[ 
	w_j + \sum_{(j, j') \in E} \dualSc_{s,j \to j'} - \sum_{(j', j) \in E} \dualSc_{s,j' \to j} = w(S(j)).
\]

Thus, the left side of~\eqref{eq:dual-red2} is equal to~$\sum_{s = t^*_j}^{C_j} w(S(j))$.

The facts that~$W(t_1) \geq W(t_2)$ at any~$t_1 \leq t_2$ and that~$j$ is processed by~$R_{j}^{t'}$ units at any time~$t' \in [t_j^*, C_j]$  combined with the assumption $\frac{w(S(j))}{W(t)} \leq \rho \cdot R_j^t$ imply the following:
\begin{equation*}
    \sum_{s = t^*_j}^{C_j} \frac{w(S(j))}{W(t)} \leq \sum_{s = t^*_j}^{C_j} \frac{w(S(j))}{W(s)} \leq \sum_{s = t^*_j}^{C_j} \rho \cdot R_j^s \leq \rho \cdot p_j.
\end{equation*}
Rearranging it, using the definition of~$\dualSb_t$ and~$\optspeed \geq \kappa \rho$ gives
  \[
    \sum_{s = t^*_j}^{C_j} w(S(j)) \leq \rho \cdot  p_j \cdot W(t) = \rho \kappa \cdot p_j \cdot \dualSb_t \leq \optspeed \cdot p_j \cdot \dualSb_t,
  \]
  which implies \eqref{eq:dual-red2} and thus proves the statement.
\end{proof}

\begin{proof}[Proof of \Cref{thm:rates-algo}]
We set~$\optspeed = \rho \kappa$. Weak LP duality, \Cref{lemma:chain-robin-dual-feasible}, and \Cref{lemma:chain-robin-dual-objective} imply
\begin{equation*} 
    {\rho \kappa} \cdot \opt \geq \opt_{\rho \kappa} \geq \sum_{j} \dualSa_j - \sum_{t} \dualSb_{t} = \left( 1 - \frac{1}{\kappa} \right) \cdot \alg.
\end{equation*}

Choosing~$\kappa = 2$, we conclude that~$\alg \leq 4 \rho \cdot \opt$.
\end{proof}

\subsection{Learning-Augmented Algorithms}
\label{sec:chains-error}

In this section, we extend the algorithms presented in \Cref{sec:chains-consistent} to achieve a smooth error-dependency in the case of inaccurate predictions, while preserving constant consistency.
Further, we use the time-sharing technique (cf. \Cref{sec:time-sharing}) to ensure a robustness of $\bigO(\numc)$.

\paragraph{Static Weight Predictions for Chains}

Here, the main challenges are as follows: we only have access to the potentially wrong predictions~$\hW_c$ on the total chain weight for all~$c \in \chains$ and, therefore, we execute~\Cref{alg:chain-robin} using~$\hW_c$ instead of~$W_c$. In particular, the weight of a chain~$c$ might be \emph{underpredicted},~$\hW_c < W_c$, or \emph{overpredicted},~$\hW_c > W_c$. This means that~$\sum_c \hW_c$ may not be the accurate total weight of the instance and that the recomputation of~$W_c(t)$ in Line~$5$ may be inaccurate. In \Cref{app:chains-error-dep}, we show how to encode the error due to underpredicted chains in an instance~$\chains_u$ and the error due to overpredicted chains in an instance~$\chains_o$, similar to an error-dependency proposed in~\cite{BamasMS20} for online set cover. We prove the following result:

\begin{restatable}{theorem}{ThmStaticWeightsError}
	\label{thm:ThmStaticWeightsError}
	For minimizing the total weighted completion time of jobs with online chain precedence constraints on a single machine, 
	there is a non-clairvoyant algorithm with predicted chain weights with a competitive ratio of at most
	\[
		\mathcal{O}(1) \cdot \min\left\{ 1 + \frac{\opt(\chains_o) + \numc \cdot \opt(\chains_u)}{\opt}, \numc \right\}.
	\] 
\end{restatable}

\paragraph{Adaptive Weight Predictions for Out-Forests}
To capture the quality of an adaptive prediction, we intuitively need to measure its quality over the whole execution. To this end, we use the maximal distortion factor of the weight predictions of every possible front job, which in fact can be any job in $J$. We prove in \Cref{app:adaptive-weight-out-forests}:

\begin{restatable}{theorem}{thmAdaptiveWeightsError}
	For minimizing the total weighted completion time on a single machine with online out-forest precedence constraint and adaptive weight predictions, there is a non-clairvoyant algorithm with a competitive ratio of at most
	\[
		\mathcal{O}(1) \cdot \min\left\{\max_{v \in J} \frac{\hW_{v}}{w(S(v))} \cdot \max_{v \in J} \frac{w(S(v))}{\hW_{v}}, \omega \right\}.
	\]
\end{restatable}

\section{Weight Order Predictions}

\label{sec:weight-order}
We consider static and adaptive weight order predictions. As strong lower bounds hold for in-trees, even for the more powerful adaptive weight predictions (cf.~\Cref{lb:in-trees}), we focus on chains and out-forest precedence constraints.

Further, we introduce an error measure for wrongly predicted orders.
A natural function on orders is the \emph{largest inversion}, i.e., the maximum distance between the position of a front job in an order prediction~$\hpreceq_t$ and the true order~$\preceq_t$. 
However, if all out-trees have almost the same weight, just perturbed by some small constant, this function indicates a large error for the reverse order, although it will arguably perform nearly as good as the true order.
To mitigate this overestimation, we first introduce~$\epsilon$-approximate inversions. Formally, for every precision constant~$\epsilon > 0$, we define
\[
    \largestinv(\epsilon) = \max_{t, j \in F_t} \left| \left\{ i \in F_t \bigg| \frac{w(S(j))}{1 + \epsilon} \geq w(S(i)) \land i \hpreceq_t j \right\} \right|.
\]
Note that~$\largestinv(\epsilon) \geq 1$ for every~$\epsilon > 0$, because~$\hpreceq_t$ is reflexive.
We define the \emph{$\epsilon$-approximate largest inversion} error as~\(
    \max\{1 + \epsilon, \largestinv(\epsilon)\}.
\)
We show performance guarantees depending on this error which hold for any~$\epsilon > 0$.
Therefore, we intuitively get a pareto frontier between the precision~$(1+\epsilon)$ and~$\largestinv(\epsilon)$, the largest distance of inversions which are worse than the precision.
A configurable error with such properties has been applied to other learning-augmented algorithms~\cite{AzarPT22,Bernardini22universal}.

\subsection{Adaptive Weight Order}
We introduce Algorithm~\ref{alg:weight-order}, which exploits access to the adaptive order~$\hpreceq_t$.  
In a sense, the idea of the algorithm is to emulate Algorithm~\ref{alg:chain-robin} for weight predictions. Instead of having access to the total remaining weight of every out-tree to computing rates, Algorithm~\ref{alg:weight-order} uses~$\hpreceq_t$ to approximate the rates. For every~ front job $j \in F_t$, let~$i_j$ be the position of $j$ in~$\hpreceq_t$.
Recall that~$H_k$ denotes the~$k$th harmonic number.

\begin{algorithm}[tb]
	\caption{Adaptive weight order algorithm}\label{alg:weight-order}
	\begin{algorithmic}[1]
		\REQUIRE Time~$t$, front jobs $F_t$, adaptive order~$\hpreceq_t$.
		\STATE Process every $j \in F_t$ with rate~$(H_{\abs{F_t}} \cdot i_j)^{-1}$, where $i_j$ is the position of $j$ in~$\hpreceq_t$.
	\end{algorithmic}
\end{algorithm}

\begin{theorem}\label{thm:adaptive-weight-order}
	For any $\epsilon > 0$, \Cref{alg:weight-order} has a competitive ratio of at most~$4 H_\numc \cdot \max\{1 + \epsilon, \largestinv(\epsilon)\}$ for minimizing the total weighted completion time on a single machine with online out-forest precedence constraints. \end{theorem}

\begin{proof}
	We first observe that the rates of the algorithm are feasible, because~$\sum_{j \in F_t} \frac{1}{H_{\abs{F_t}} \cdot i_j} = \frac{H_{\abs{F_t}}}{H_{\abs{F_t}}} = 1$. 

	Fix a time~$t$ and an $\epsilon > 0$. Assume that~$j_1 \hpreceq_t \ldots \hpreceq_t j_{\abs{F_t}}$, and fix a front job~$j_i \in F_t$.
	The algorithm processes~$j_i$ at time~$t$ with rate 
	\(
		R_{j_i}^t = (H_{\abs{F_t}} \cdot i)^{-1} \geq (H_{\numc} \cdot i)^{-1}.
	\) 
	Note that showing~$\frac{w(S(j_i))}{W(t)} \le H_\numc \cdot \max\{1 + \epsilon, \largestinv(\epsilon)\} \cdot L^t_{j_i}$ implies the theorem via~\Cref{thm:rates-algo}.
	Assume otherwise, i.e.,~$\frac{w(S(j_i))}{W(t)} > \frac{1}{i} \cdot \max\{1 + \epsilon, \largestinv(\epsilon)\}$. 
	For the sake of readability, we define $K_{>} = \{k \in [i-1] \mid w(S(i_k)) > \frac{w(S(j_i))}{1+\epsilon} \}$ and $K_{\leq} = \{k \in [i] \mid w(S(i_k)) \leq \frac{w(S(j_i))}{1+\epsilon} \}$.
    Since in an out-forest the sets $S(j)$ are pairwise disjoint for all front jobs $j \in F_t$,
    \begin{align*}     
        1 \ge \sum_{k \in [i]} \frac{w(S(i_k))}{W(t)} 
        \ge \sum_{k \in K_{>}} \frac{w(S(i_k))}{W(t)} + \sum_{k \in K_{\leq}} \frac{w(S(i_k))}{W(t)}.
    \end{align*}
    Consider the second sum. First, observe that this sum has at most~$\largestinv(\epsilon)$ many terms, including the one for $j_i$, and that each such term is at most $\frac{w(S(j_i))}{W(t)}$.
    Then, observe that every term in the first sum is at least~$\frac{w(S(j_i))}{(1+\epsilon)W(t)}$.
    Thus, we can further lower bound the sum of the two sums by 
    \begin{align*}     
        &\frac{1}{1 + \epsilon}\sum_{k \in K_{>}} \frac{w(S(j_i))}{W(t)} + \frac{1}{\largestinv(\epsilon)}\sum_{k \in K_{\leq}} \frac{w(S(j_i))}{W(t)} \\
        &\geq\frac{1}{\max\{1 + \epsilon, \largestinv(\epsilon)\}} \sum_{k \in [i]} \frac{w(S(j_i))}{W(t)} > \sum_{k=1}^i \frac{1}{i} = 1.
    \end{align*}
	This is a contradiction.
\end{proof}
Using this theorem, we conclude the following corollary.

\begin{corollary}
There exists a non-clairvoyant weight-oblivious algorithm for the problem of minimizing the total weighted completion time of~$n$ jobs on a single machine with a competitive ratio of at most~$\bigO(\log n)$ when given access to the order of the job's weights.
\end{corollary}

\subsection{Static Weight Order} 

If we only have access to~$\hpreceq_0$, a natural approach would be to compute the initial rates as used in Algorithm~\ref{alg:weight-order} and just not update them. 
As \Cref{obs:lb-trees-static} rules out well-performing algorithms for out-trees, we focus on chains.
Even for chains, we show that this algorithm has a competitive ratio of at least~$\Omega(\numc \cdot H_\numc)$.

\begin{restatable}{lem}{LBStaticWO}
	\label{lem:LemLBStaticWO}
	The variant of Algorithm~\ref{alg:weight-order} that computes the rates using~$\hpreceq_0$ instead of~$\hpreceq_t$ is at least~$\Omega(\numc \cdot H_{\numc})$-competitive, even if~$\hpreceq_0$ equals $\preceq_0$.
\end{restatable}

\begin{proof}
	Consider an instance with~$\numc$ chains, each with a total weight of one. Then,~$\preceq_0$ is just an arbitrary order of the chains. Recall that the algorithm starts processing the chains~$c$ with rate~$(H_\numc \cdot i_c)^{-1}$, where $i_c$ is the position of $c$ in the order. We define the first~$\numc-1$ chains to have their total weight of one at the very first job and afterwards only jobs of weight zero. Chain~$\numc$, the slowest chain, has its total weight on the last job. 
	We define the chains~$c$ to contain a total of~$d \cdot H_\numc \cdot i_c$ jobs with unit processing times, for some common integer~$d$. This means that the algorithm finishes all chains at the same time.
	The optimal solution value for this instance is~$\numc \cdot (\numc + 1) + \numc - 1 + d \cdot H_\numc \cdot \numc$, where~$\numc \cdot (\numc + 1)$ is the optimal sum of completion times for the first~$\numc-1$ chains, ~$d \cdot H_\numc \cdot \numc$ is the cost for processing the last chain, and~$\numc - 1$ is the cost for delaying the last chains by the~$\numc - 1$ time units needed to process the first jobs of the first~$\numc-1$ chains. The solution value of the algorithm is at least~$d \cdot H_\numc^2 \cdot \numc^2$ as this is the cost for just processing the last chain. Thus, for large~$d$, the competitive ratio tends to~$H_\numc \cdot \numc$.
\end{proof}

However, the lower bound instance of the lemma requires~$\numc$ to be \enquote{small} compared to the number of jobs, in case of unit jobs, or to~$P := \sum_j p_j$, otherwise.
We exploit this to prove the following theorem in \Cref{app:weight-order}.

\begin{restatable}{theorem}{ThmStaticOrder}
		For any $\epsilon > 0$, Algorithm~\ref{alg:weight-order} has a competitive ratio of at most~$\mathcal{O}(H_\numc^2 \sqrt{P} \cdot \max\{1 + \epsilon, \largestinv(\epsilon)\})$ when computing rates with~$\hpreceq_0$ instead of~$\hpreceq_t$ at any time $t$.
		For unit jobs, it is~$\mathcal{O}(H_\numc^2 \sqrt{n} \cdot \max\{1 + \epsilon, \largestinv(\epsilon)\})$-competitive.
\end{restatable}

\section{Average Predictions}
	Recall that an average predictions give access to predicted values $\averagep_{v}$ on~$\average(S(v)) = (\sum_{u \in S(v)} w_u)/(\sum_{u \in S(v)} p_u)$ for each~$v \in \frontjobs_t$. 
	We show the following lower bound for chains with unit jobs, where average predictions coincide with the average weight of the jobs in the respective chain.
	The lower bound exploits that we can append jobs of weight zero to a chain in order to manipulate the average weight of the chain until all chains have the same average. 

	\begin{restatable}{lem}{LBChainsAverage}\label{l:lb:chains}
    Any algorithm which has only access to correct adaptive average predictions is at least $\Omega(\sqrt{n})$-competitive even for chain precedence constraints with unit jobs.
    \end{restatable}	

	\begin{proof}
		Consider an instance composed of $\sqrt{n} \in \mathbb{N}$ chains of unit jobs, where the first two jobs of the first chain have weights 1 resp. $n - \sqrt{n}$, followed by $n - \sqrt{n} - 1$ zero weight jobs. The other $\sqrt{n}-1$ chains are single jobs with weight $1$. For an algorithm, all chains look identical since the first jobs have weight $1$ and the average of every chain is equal to $1$. Therefore, an adversary can ensure that the algorithm processes the first chain last, giving an objective value of $\sum^{\sqrt{n}}_{i=1} i + (\sqrt{n}+1)(n - \sqrt{n}) = \Omega(n\sqrt{n})$, while a solution which schedules the heavy weight job initially achieves an objective value of at most $1 + 2(n - \sqrt{n}) + \sum^{\sqrt{n} - 1}_{i=1} (3 + i) = \bigO(n)$.
		The adaptivity of the predictions does not help for this lower bound as the algorithm would only receive meaningful updates once it finishes the first job of the first chain, which is too late.
	\end{proof}

\section{Final Remarks}

We initiated the study of learning-augmented algorithms for scheduling with online precedence constraints by considering a hierarchy of prediction models based on their entropy. For several models of the hierarchy, we were able to show that the predicted information is sufficient to break lower bounds for algorithms without predictions. We hope that our approach leads to more discussions on the identification of the \enquote{right} prediction model in learning-augmented algorithm design. As a next research step, we suggest investigating the missing bounds for our prediction models, e.g., an upper bound for average predictions.

\bibliographystyle{alpha}
\bibliography{../literature}

\newcommand{\etalchar}[1]{$^{#1}$}
\begin{thebibliography}{BLMS{\etalchar{+}}22}

\bibitem[ACE{\etalchar{+}}20]{AntoniadisCE0S20}
Antonios Antoniadis, Christian Coester, Marek Eli{\'{a}}s, Adam Polak, and
  Bertrand Simon.
\newblock Online metric algorithms with untrusted predictions.
\newblock In {\em {ICML}}, volume 119 of {\em Proceedings of Machine Learning
  Research}, pages 345--355. {PMLR}, 2020.

\bibitem[AGKP22]{Anand0KP22}
Keerti Anand, Rong Ge, Amit Kumar, and Debmalya Panigrahi.
\newblock Online algorithms with multiple predictions.
\newblock In {\em {ICML}}, volume 162 of {\em Proceedings of Machine Learning
  Research}, pages 582--598. {PMLR}, 2022.

\bibitem[AGS22]{AntoniadisGS22}
Antonios Antoniadis, Peyman~Jabbarzade Ganje, and Golnoosh Shahkarami.
\newblock A novel prediction setup for online speed-scaling.
\newblock In {\em {SWAT}}, volume 227 of {\em LIPIcs}, pages 9:1--9:20. Schloss
  Dagstuhl - Leibniz-Zentrum f{\"{u}}r Informatik, 2022.

\bibitem[All70]{Allen1970control}
Frances~E Allen.
\newblock Control flow analysis.
\newblock {\em ACM Sigplan Notices}, 5(7):1--19, 1970.

\bibitem[{ALP}23]{alps}
{ALPS contributors}.
\newblock Algorithms with predictions paper-tracker, 2023.

\bibitem[ALT21]{AzarLT21}
Yossi Azar, Stefano Leonardi, and Noam Touitou.
\newblock Flow time scheduling with uncertain processing time.
\newblock In {\em {STOC}}, pages 1070--1080. {ACM}, 2021.

\bibitem[ALT22]{AzarLT22}
Yossi Azar, Stefano Leonardi, and Noam Touitou.
\newblock Distortion-oblivious algorithms for minimizing flow time.
\newblock In {\em {SODA}}, pages 252--274. {SIAM}, 2022.

\bibitem[AM09]{AmbuhlM09}
Christoph Amb{\"{u}}hl and Monaldo Mastrolilli.
\newblock Single machine precedence constrained scheduling is a vertex cover
  problem.
\newblock {\em Algorithmica}, 53(4):488--503, 2009.

\bibitem[APT22]{AzarPT22}
Yossi Azar, Debmalya Panigrahi, and Noam Touitou.
\newblock Online graph algorithms with predictions.
\newblock In {\em {SODA}}, pages 35--66. {SIAM}, 2022.

\bibitem[BBEM12]{BeaumontBEM12}
Olivier Beaumont, Nicolas Bonichon, Lionel Eyraud{-}Dubois, and Loris Marchal.
\newblock Minimizing weighted mean completion time for malleable tasks
  scheduling.
\newblock In {\em {IPDPS}}, pages 273--284. {IEEE} Computer Society, 2012.

\bibitem[BDK{\etalchar{+}}22]{BampisDKLP22}
Evripidis Bampis, Konstantinos Dogeas, Alexander~V. Kononov, Giorgio Lucarelli,
  and Fanny Pascual.
\newblock Scheduling with untrusted predictions.
\newblock In {\em {IJCAI}}, pages 4581--4587. ijcai.org, 2022.

\bibitem[BE98]{Borodin98}
Allan Borodin and Ran El{-}Yaniv.
\newblock {\em Online computation and competitive analysis}.
\newblock Cambridge University Press, 1998.

\bibitem[BKL21]{BienkowskiKL21}
Marcin Bienkowski, Artur Kraska, and Hsiang{-}Hsuan Liu.
\newblock Traveling repairperson, unrelated machines, and other stories about
  average completion times.
\newblock In {\em {ICALP}}, volume 198 of {\em LIPIcs}, pages 28:1--28:20.
  Schloss Dagstuhl - Leibniz-Zentrum f{\"{u}}r Informatik, 2021.

\bibitem[BLMS{\etalchar{+}}22]{Bernardini22universal}
Giulia Bernardini, Alexander Lindermayr, Alberto Marchetti-Spaccamela, Nicole
  Megow, Leen Stougie, and Michelle Sweering.
\newblock A universal error measure for input predictions applied to online
  graph problems.
\newblock In {\em NeurIPS}, 2022.

\bibitem[BMS20]{BamasMS20}
{\'{E}}tienne Bamas, Andreas Maggiori, and Ola Svensson.
\newblock The primal-dual method for learning augmented algorithms.
\newblock In {\em NeurIPS}, 2020.

\bibitem[CHS22]{ChoHS2022}
Woo{-}Hyung Cho, Shane~G. Henderson, and David~B. Shmoys.
\newblock Scheduling with predictions.
\newblock {\em CoRR}, abs/2212.10433, 2022.

\bibitem[CPS{\etalchar{+}}96]{ChakrabartiPSSSW96}
Soumen Chakrabarti, Cynthia~A. Phillips, Andreas~S. Schulz, David~B. Shmoys,
  Clifford Stein, and Joel Wein.
\newblock Improved scheduling algorithms for minsum criteria.
\newblock In {\em {ICALP}}, volume 1099 of {\em Lecture Notes in Computer
  Science}, pages 646--657. Springer, 1996.

\bibitem[CS05]{CorreaS05}
Jos{\'{e}}~R. Correa and Andreas~S. Schulz.
\newblock Single-machine scheduling with precedence constraints.
\newblock {\em Math. Oper. Res.}, 30(4):1005--1021, 2005.

\bibitem[DIL{\etalchar{+}}22]{DinitzILMV22portfolios}
Michael Dinitz, Sungjin Im, Thomas Lavastida, Benjamin Moseley, and Sergei
  Vassilvitskii.
\newblock Algorithms with prediction portfolios.
\newblock In Alice~H. Oh, Alekh Agarwal, Danielle Belgrave, and Kyunghyun Cho,
  editors, {\em Advances in Neural Information Processing Systems}, 2022.

\bibitem[DLY91]{duLY91}
Jianzhong Du, Joseph~Y.{-}T. Leung, and Gilbert~H. Young.
\newblock Scheduling chain-structured tasks to minimize makespan and mean flow
  time.
\newblock {\em Inf. Comput.}, 92(2):219--236, 1991.

\bibitem[EdLMS22]{ErlebachLMS22}
Thomas Erlebach, Murilo~Santos de~Lima, Nicole Megow, and Jens Schl{\"{o}}ter.
\newblock Learning-augmented query policies for minimum spanning tree with
  uncertainty.
\newblock In {\em {ESA}}, volume 244 of {\em LIPIcs}, pages 49:1--49:18.
  Schloss Dagstuhl - Leibniz-Zentrum f{\"{u}}r Informatik, 2022.

\bibitem[ELM{\etalchar{+}}22]{EberleLMNS22}
Franziska Eberle, Alexander Lindermayr, Nicole Megow, Lukas N{\"{o}}lke, and
  Jens Schl{\"{o}}ter.
\newblock Robustification of online graph exploration methods.
\newblock In {\em {AAAI}}, pages 9732--9740. {AAAI} Press, 2022.

\bibitem[GGKS19]{GargGKS19}
Naveen Garg, Anupam Gupta, Amit Kumar, and Sahil Singla.
\newblock Non-clairvoyant precedence constrained scheduling.
\newblock In {\em {ICALP}}, volume 132 of {\em LIPIcs}, pages 63:1--63:14.
  Schloss Dagstuhl - Leibniz-Zentrum f{\"{u}}r Informatik, 2019.

\bibitem[Hay08]{Hayes08}
Brian Hayes.
\newblock Cloud computing.
\newblock {\em Commun. {ACM}}, 51(7):9--11, 2008.

\bibitem[HSSW97]{hallSSW97}
Leslie~A. Hall, Andreas~S. Schulz, David~B. Shmoys, and Joel Wein.
\newblock Scheduling to minimize average completion time: Off-line and on-line
  approximation algorithms.
\newblock {\em Math. Oper. Res.}, 22(3):513--544, 1997.

\bibitem[IKQP21]{Im0QP21}
Sungjin Im, Ravi Kumar, Mahshid~Montazer Qaem, and Manish Purohit.
\newblock Non-clairvoyant scheduling with predictions.
\newblock In {\em {SPAA}}, pages 285--294. {ACM}, 2021.

\bibitem[J{\"{a}}g21]{Jager21}
Sven~Joachim J{\"{a}}ger.
\newblock {\em Approximation in deterministic and stochastic machine
  scheduling}.
\newblock PhD thesis, Technical University of Berlin, Germany, 2021.

\bibitem[JM22]{JinM22}
Billy Jin and Will Ma.
\newblock Online bipartite matching with advice: Tight robustness-consistency
  tradeoffs for the two-stage model.
\newblock In Alice~H. Oh, Alekh Agarwal, Danielle Belgrave, and Kyunghyun Cho,
  editors, {\em Advances in Neural Information Processing Systems}, 2022.

\bibitem[KC03]{KimC03a}
Jae{-}Hoon Kim and Kyung{-}Yong Chwa.
\newblock Non-clairvoyant scheduling for weighted flow time.
\newblock {\em Inf. Process. Lett.}, 87(1):31--37, 2003.

\bibitem[Law78]{lawler78}
Eugene~L. Lawler.
\newblock Sequencing jobs to minimize total weighted completion time subject to
  precedence constraints.
\newblock {\em Annals of Discrete Mathematics}, 2:75--90, 1978.

\bibitem[LK78]{lenstrR78}
Jan~Karel Lenstra and A.~H. G.~Rinnooy Kan.
\newblock Complexity of scheduling under precedence constraints.
\newblock {\em Oper. Res.}, 26(1):22--35, 1978.

\bibitem[LM22]{LindermayrM22}
Alexander Lindermayr and Nicole Megow.
\newblock Permutation predictions for non-clairvoyant scheduling.
\newblock In {\em {SPAA}}, pages 357--368. {ACM}, 2022.

\bibitem[LMS22]{LindermayrMS22}
Alexander Lindermayr, Nicole Megow, and Bertrand Simon.
\newblock Double coverage with machine-learned advice.
\newblock In {\em {ITCS}}, volume 215 of {\em LIPIcs}, pages 99:1--99:18.
  Schloss Dagstuhl - Leibniz-Zentrum f{\"{u}}r Informatik, 2022.

\bibitem[LRLE17]{LynnRLE17}
Theo Lynn, Pierangelo Rosati, Arnaud Lejeune, and Vincent~C. Emeakaroha.
\newblock A preliminary review of enterprise serverless cloud computing
  (function-as-a-service) platforms.
\newblock In {\em CloudCom}, pages 162--169. {IEEE} Computer Society, 2017.

\bibitem[McN59]{mcnaughton1959scheduling}
Robert McNaughton.
\newblock Scheduling with deadlines and loss functions.
\newblock {\em Management science}, 6(1):1--12, 1959.

\bibitem[MPT94]{motwani1994nonclairvoyant}
Rajeev Motwani, Steven Phillips, and Eric Torng.
\newblock Nonclairvoyant scheduling.
\newblock {\em Theoretical computer science}, 130(1):17--47, 1994.

\bibitem[MV20]{MitzenmacherV20}
Michael Mitzenmacher and Sergei Vassilvitskii.
\newblock Algorithms with predictions.
\newblock In {\em Beyond the Worst-Case Analysis of Algorithms}, pages
  646--662. Cambridge University Press, 2020.

\bibitem[MV22]{MitzenmacherV22}
Michael Mitzenmacher and Sergei Vassilvitskii.
\newblock Algorithms with predictions.
\newblock {\em Commun. {ACM}}, 65(7):33--35, 2022.

\bibitem[PSK18]{PurohitSK18}
Manish Purohit, Zoya Svitkina, and Ravi Kumar.
\newblock Improving online algorithms via {ML} predictions.
\newblock In {\em NeurIPS}, pages 9684--9693, 2018.

\bibitem[SBW19]{ShahradBW19}
Mohammad Shahrad, Jonathan Balkind, and David Wentzlaff.
\newblock Architectural implications of function-as-a-service computing.
\newblock In {\em {MICRO}}, pages 1063--1075. {ACM}, 2019.

\bibitem[Tim03]{timkovsky03}
Vadim~G. Timkovsky.
\newblock Identical parallel machines vs. unit-time shops and preemptions vs.
  chains in scheduling complexity.
\newblock {\em Eur. J. Oper. Res.}, 149(2):355--376, 2003.

\end{thebibliography}

\newpage
\appendix

\section{Input Predictions}\label{app:input}

This section is devoted to the proof of \Cref{thm:input-predictions}.

Let~$\hchains$ denote the set of predicted chains, and let $\hw_j$ denote the predicted weight of a job $j$ of the instance. All processing requirements are equal to $1$, and the algorithm is aware of this.
We assume w.l.o.g. that~$\abs{\hchains} = \abs{\chains}$ by adding chains with zero (predicted) weight, and that predicted and actual chains have the same identities. That is, there exists exactly one predicted chain~$c_i \in \hchains$ for each actual chain~$\hc_i \in \chains$, which an algorithm can match to each other. 

Our error measure further requires that there exists for every actual job a predicted counterpart, and vice versa. 
For a chain $c$ let $|c|$ denote the number of jobs of chain $c$.
We define augmentations of~$\chains$ and~$\hchains$ as follows. Let $\chains'$ be composed of all jobs of~$\chains$, and additionally, for every paired chains~$c_i \in \chains$ and~$\hc_i \in \hchains$:
\begin{itemize}
	\item~if~$\abs{\hc_i} > \abs{c_i}$, we add~$\abs{\hc_i} - \abs{c_i}$ jobs~$J_u$ with weight~$0$ at the end of~$c_i$ in~$\chains'$. Note that~$\opt(\chains) = \opt(\chains')$.
	\item~if~$\abs{c_i} > \abs{\hc_i}$, we add~$\abs{c_i} - \abs{\hc_i}$ jobs~$J_a$ with predicted weight~$0$ at the end of~$\hc_i$ in~$\hchains'$. 
\end{itemize}
Note that this construction ensures $\opt(\hchains) = \opt(\hchains')$.

For the sake of analysis, assume w.l.o.g. that both~$\chains'$ and~$\hchains'$ share the same set of jobs~$J'$. Let~$n' = \abs{J'}$.
We define~$\opt(\{w'_j\}_{j})$ as the objective of an optimal solution for~$J'$ where a job~$j$ has weight~$w'_j$. 
We further define
\[ 
	\opt(\{w'_j\}_{j}, \{w_j\}_{j}) = \max\left\{\sum_{j \in J'} w'_j C^*_j \mid \{C^*_j\}_j \text{ is an optimal schedule for } \{w_j\}_{j} \right\}.
\]

Given two fixed augmented instances~$\chains'$ and~$\hchains'$, we define the input prediction error~$\Lambda = \Gamma_u + \Gamma_a$:
\begin{itemize}
	\item a job~$j \in J'$ has \emph{unexpected actual} weight if~$w_j > \hw_j$. The prediction error due to all unexpected weights can be expressed as~$\Gamma_u = \opt(\{\max\{\hw_j, w_j\} - w_j \}_{j}, \{w_j\}_{j})$.
	\item a job~$j \in J'$ has \emph{absent predicted} weight if~$\hw_j > w_j$. The prediction error due to all absent weights can be expressed as~$\Gamma_a = \opt(\{\max\{w_j, \hw_j\} - \hw_j\}_{j}, \{\hw_j\}_{j})$.
\end{itemize}

\thmInputPredErrorDep*

Recall that $\bigO(\omega)$-robustness can be achieved via \Cref{lemma:chain-robust,thm:alg-combination}. Thus, we can prove the theorem by deriving an $\mathcal{O}(1+\Lambda)$-competitive algorithm.
\begin{proof}
    We analyze the following algorithm:
    \begin{enumerate}[1)]
        \item Efficiently compute an optimal solution based on $\hchains$~\cite{lawler78}. This yields an non-preemptive schedule for the predicted instance, i.e., an order of the jobs.
        \item Follow the computed solution. The following situations might occur:
        \begin{enumerate}[a)]
            \item a chain finishes earlier than expected. In this case, discard the remaining predicted jobs of this chain in the precomputed schedule. 
            \item a chain continues although there are no more jobs in this chain in the algorithms schedule. In this case, schedule the remaining jobs in an arbitrary order at the end of the precomputed schedule.
        \end{enumerate}
    \end{enumerate}
    Let $\alg$ denote the objective value of this algorithm. We first observe that $\alg \leq \opt(\{w_j\}_{j}, \{\hw_j\}_{j})$. To see this, recall that the algorithm first follows an optimal schedule for jobs $J' \setminus J_u$ and then schedules all unexpected jobs $J_u$ at the end due to case b). 
    Since jobs $J_u$ have predicted weight $0$ in $\hchains'$, we can assume that an optimal solution for $\hchains'$ first schedules jobs $J' \setminus J_u$ as our algorithm with the same objective value and makespan as our algorithm, and then schedules jobs $J_u$ in any order. Since $\opt(\{w_j\}_{j}, \{\hw_j\}_{j})$ is an upper bound on the actual objective for any such order, the inequality follows.
    It further holds that
    \begin{align*}
    \opt(\{w_j\}_{j}, \{\hw_j\}_{j}) 
    &\leq \opt(\{\max\{w_j, \hw_j\}\}_{j}, \{\hw_j\}_{j}) \\
    &= \opt(\{\hw_j\}_{j}) + \opt(\{\max\{w_j, \hw_j\} - \hw_j\}_{j}, \{\hw_j\}_{j}) \\
    &\leq \opt(\{\hw_j\}_{j}, \{w_j\}_{j}) + \opt(\{\max\{w_j, \hw_j\} - \hw_j\}_{j}, \{\hw_j\}_{j})  \\
    &\leq \opt(\{\max\{\hw_j, w_j\}\}_{j}, \{w_j\}_{j}) + \opt(\{\max\{w_j, \hw_j\} - \hw_j\}_{j}, \{\hw_j\}_{j})  \\
    &\leq \opt(\{w_j\}_{j}) + \opt(\{\max\{\hw_j, w_j\} - w_j \}_{j}, \{w_j\}_{j}) + \opt(\{\max\{w_j, \hw_j\} - \hw_j\}_{j}, \{\hw_j\}_{j}) \\
    &= \opt(\{w_j\}_{j}) + \Lambda .
    \end{align*}
    We finally observe that $\opt(\{w_j\}_{j}) =  \opt(\chains)$, as jobs $J_a$ do not influence the objective value of an optimal solution.
\end{proof}

\section{Weight Value Predictions}

\subsection{A Learning-Augmented Algorithm for Chains}
\label{app:chains-error-dep}

We give a formal proof of the following theorem.

\ThmStaticWeightsError*

In order to give the proof, we first formally define the \emph{predicted instance} (including $\chains_o$ and $\chains_u$).

\begin{definition}[predicted instances]
	The \emph{predicted instance}~$\hchains$, the \emph{underpredicted subinstance}~$\chains_u$, and the \emph{overpredicted subinstance}~$\chains_o$ are constructed by considering
	for every~$c = [j_1,\ldots,j_\ell] \in \chains$ the following cases:
	\begin{enumerate}[(i)]
		\item if~$\hW_c = W_c$, then the chain~$\hc = c$ with job weights~$\hw_{j} = w_j$ for all~$j \in c$ is added to~$\hchains$.
		\item if~$\hW_c < W_c$, then the chain~$\hc = [j_1,\ldots,j_k]$, where~$k$ is the smallest index s.t.~$\hW_c \leq \sum_{i=1}^{k} w_i$, with weights~$\hw_{j_i} = w_{j_i}$ for all~$1 \leq i \leq k - 1$ and~$\hw_{j_k} = \hW_c - \sum_{i=1}^{k-1} w_i$ is added to~$\hchains$. Additionally, a chain~$c_u = [\bot, j_{k+1}, \ldots, j_\ell]$ with weights~$\hw_{j_i} = w_{j_i}$ for all~$k+1 \leq i \leq \ell$ and~$\hw_{\bot} = \sum_{i=1}^{k} w_i - \hW_c$ is added to~$\chains_u$, where the processing requirement of~$\bot$ is equal to the total processing requirement of~$\hc$.
		\item if~$\hW_c > W_c$, then chain~$\hc = [j_1,\ldots,j_\ell]$ with weights $\hw_{j_i} = w_{j_i}$ for all~$1 \leq i \leq \ell - 1$ and $\hw_{j_\ell} = \hW_c - \sum_{i=1}^{\ell-1} w_i$ is added to~$\hchains$. Additionally, a chain~$c_o = [\top]$ is added to~$\chains_o$, where the weight of $\top$ is equal to~$\hW_c - \sum_{i=1}^{\ell} w_i$ and its processing requirement is equal to the total processing requirement of~$\hc$.
	\end{enumerate}
	
	Finally,~$\chains_p$ is a copy of~$\hchains$ where for every overpredicted chain~$\hc = [j_1,\ldots,j_\ell] \in \hchains$ the weight of its last job~$j_\ell$ is set to~$w_{j_\ell}$, the weight of the job in the actual instance. This weight is strictly smaller than the weight $\hw_{j_\ell} = \hW_c - \sum_{i=1}^{\ell-1} w_i$ of the job in instance $\hchains$.
\end{definition}

Note that for every~$\hc \in \hchains$, we have~$\sum_{j \in \hc} \hw_j = \hW_c$, i.e., the predicted weights are correct for~$\hchains$.
In the following, we nevertheless call a chain~$\hc \in \hchains$ overpredicted resp. underpredicted if that is true for its corresponding chain in~$\chains$.
Since every job~$j$ of~$\chains_p$ is also part of~$\chains$ with the same processing requirement and a weight of at most $w_j$ and the chains in $\chains_p$ are prefixes of the chains in $\chains$, we conclude: 

\begin{proposition}
	$\opt(\chains_p) \leq \opt(\chains)$.
\end{proposition}

We use the algorithm of~\Cref{lemma:chain-robust} in combination with the times-haring of~\Cref{thm:alg-combination} to define~\Cref{alg:chain-robin-robust}.

\begin{algorithm}[t]
	\caption{Learning-augmented WRR on Chains}\label{alg:chain-robin-robust}
	\begin{algorithmic}[1]
		\STATE Execute \Cref{alg:chain-robin} with rate~$\frac{1}{2}$ using the predicted chain weights.
		\STATE Execute the algorithm of \Cref{lemma:chain-robust} with rate~$\frac{1}{2}$.
	\end{algorithmic}
\end{algorithm}

\begin{lemma}
	\Cref{alg:chain-robin-robust} with predicted chain weights achieves an objective value of at most 
	\[
	\mathcal{O}(1) \cdot \opt(\chains_p) + \mathcal{O}(1) \cdot (\opt(\chains_o) + \numc \cdot \opt(\chains_u)).
	\] 
\end{lemma}

\begin{proof}
	We first argue that~$\opt(\hchains) \leq \mathcal{O}(1) \cdot \opt(\chains_p) + \mathcal{O}(1) \cdot \opt(\chains_o)$. 
	To this end, consider the instance~$\chains_p \cup \chains_o$. 
	Every correctly predicted or underpredicted chain in~$\hchains$ is contained as an identical copy in~$\chains_p$. 
	For every overpredicted chain~$\hc \in \hchains$ with weight~$\hW_\hc$ in~$\hchains$, all jobs of~$\hc$ are contained with a total weight of~$W_\hc$ in~$\chains_p$ and the remaining weight of~$\hW_\hc - W_\hc$ is contained in~$\chains_o$. 
	Additionally, it is ensured by the processing requirement of the jobs in~$\chains_o$ that their weight can only be gained when processing at least the total processing requirement of~$\hc$. 
	This implies that the time to gain weight~$\hW_\hc - W_\hc$ of every overpredicted chain~$\hc \in \hchains$ in~$\chains_p \cup \chains_o$ takes as least as long as in~$\hchains$, and thus~$\opt(\hchains) \leq \opt(\chains_p \cup \chains_o)$. 
	Finally, it is not hard to see that~$\opt(\chains_p \cup \chains_o) \le 2 \cdot \opt(\chains_p) + 2 \cdot \opt(\chains_o)$ as~$\opt(\chains_p)$ and~$\opt(\chains_o)$ can be executed in parallel by preemptively sharing the machine, yielding the claimed bound.
	
	We now show that~$\alg \le \mathcal{O}(1) \cdot (\opt(\hchains) + \numc \cdot \opt(\chains_u))$ for \Cref{alg:chain-robin-robust}, which implies the statement. 
	First consider the execution of \Cref{alg:chain-robin} in the first line. We may assume that the algorithm processes the artificial job added to each overpredicted chain in~$\hchains$, as it only increases its objective. Further, \Cref{alg:chain-robin} stops processing an underpredicted chain~$c \in \chains$ when a total weight of~$\hW_c$ has been completed on~$c$ and only finishes them at the very end of the schedule. 
	This concludes that the total objective of \Cref{alg:chain-robin} \emph{without} the weighted completion times of the jobs processed in Line~$7$ is at most~$\mathcal{O}(1) \cdot \opt(\hchains)$.
	But, due to \Cref{lemma:chain-robust} and line two of the algorithm, we conclude that \Cref{alg:chain-robin-robust} always processes such chains that are only completed in Line~$7$ of~\Cref{alg:chain-robin}  with a rate of at least~$\frac{1}{2\numc}$ and, thus, delays the completion the jobs in these chains by a factor of at most $2\numc$ compared to an optimal solution. By observing that the total weight of jobs processed in Line~$7$ by \Cref{alg:chain-robin} is exactly equal to the total weight of chains in~$\chains_u$, and the fact that the jobs in a chain~$c \in \chains_u$ can only be processed after time equal to the total processing requirement of the corresponding chain in~$\hchains$, we conclude the stated bound.
\end{proof}

\subsection{Adaptive weight predictions for out-forests}\label{app:adaptive-weight-out-forests}

\begin{algorithm}
	\caption{Weighted Round Robin on out-forests}\label{alg:chain-robin-out-trees}
	\begin{algorithmic}[1]
		\REQUIRE Out-forest~$T$ and adaptive weight predictions.
		\STATE~$t \gets 0$
		\WHILE{$F_t \neq \emptyset$}
		\STATE Process every~$v \in F_t$ with rate~$R_j^t = \frac{\hW_{j}}{\sum_{i \in F_t} \hW_{i}}$
		\STATE~$t \gets t + 1$
		\ENDWHILE
	\end{algorithmic}
\end{algorithm}

\thmAdaptiveWeightsError*

\begin{proof}

	To prove the theorem, we show that~\Cref{alg:chain-robin-out-trees} is $\mathcal{O}(\eta)$-competitive for 
	\[
		\eta = \max_{v \in J} \frac{\hW_{v}}{W_{v}} \cdot \max_{v \in J} \frac{W_{v}}{\hW_{v}}. 
	\]
	Then, \Cref{lemma:chain-robust} and~\Cref{thm:alg-combination} imply the theorem.
	
	By~\Cref{thm:rates-algo}, it suffices to show that $\frac{w(S(j))}{W(t)} \leq \eta \cdot R_j^t$ holds for any point in time $t$ during the execution of \Cref{alg:chain-robin-out-trees} and any $j \in F_t$, where $R_j^t$ is the rate with which the algorithm processes $j$ at point in time $t$. 
	Consider a fixed point in time $t$ and an arbitrary $j \in F_t$. Then, the algorithm processes $j$ with rate $R_j^t = \frac{\hW_{j}}{\sum_{i \in F_t} \hW_{i}}$ by definition. 
	We can conclude
	\begin{align*}
		\frac{w(S(j))}{W(t)} = \frac{w(S(j))}{\sum_{i \in F_t} w(S(i))} &= \frac{\hW_{j} \cdot \frac{w(S(j))}{\hW_{j}}}{\sum_{i \in F_t} \hW_{i} \cdot \frac{w(S(i))}{\hW_{i}}}\\
		&\le \frac{\left(\max_{v \in J} \frac{w(S(v))}{\hW_{v}}\right) \cdot \hW_{j}}{\left(\min_{v \in J} \frac{w(S(v))}{\hW_{v}}\right) \cdot \sum_{i \in F_t} \hW_{i}}\\
		&= \max_{v \in J} \frac{\hW_{v}}{w(S(v))} \cdot \max_{v \in J} \frac{w(S(v))}{\hW_{v}} \cdot  \frac{\hW_{j}}{\sum_{i \in F_t} \hW_{i}}\\
		&= \eta \cdot R_j^t .
	\end{align*}
	
\end{proof}

\section{Static Weight Order}\label{app:weight-order}

\ThmStaticOrder*

\begin{proof}
	Recall that $P$ denotes the sum over all job processing times in the instance.
	For a subset of chains~$S$, let~$\opt(S)$ denote the optimal objective value for the subinstance induced by~$S$. For a single chain~$c$,~$\opt(c)$ is just the cost for processing chain~$c$ with rate~$1$ on a single machine.  Clearly, $\opt(S) \ge \sum_{c \in S} \opt(c)$.
	Let~$\alg(S)$ denote the sum of weighted completion times of the jobs that belong to chains in~$S$ in the schedule computed by the algorithm.		
	
	In the first part of the proof, we assume that all chains~$c_i$ have a total processing time of at most~$\sqrt{P}$. This only decreases the objective value of~$\opt$. For~$\alg$, we will analyze the additional cost caused by longer chains afterwards.
	In a sense, we assume that~$\alg$, for each chain~$c$, has to pay all weight that appears after~$\sqrt{P}$ processing times units of the chain two times: Once artificially  after exactly~$\sqrt{P}$ time units of the chain have been processed and once at the point during the processing where the weight actually appears. This assumption clearly only increases~$\alg$.
	In the first part of the proof, we analyze only the artificial cost for such weights and ignore the actual cost. In the context of our algorithm this is equivalent to assuming the chains have total processing times of at most~$\sqrt{P}$. In the second part of the proof we will analyze the actual cost for the jobs that appear after~$\sqrt{P}$ time units in their chain.

	\paragraph{First Part.}
	Assume~$c_1 \hpreceq_0 c_2 \hpreceq_0 \ldots \hpreceq_0 c_\numc$. Therefore, the algorithm processes chain~$c_i$ with rate~$(H_\numc \cdot i)^{-1}$.
	This directly implies~$\alg(c_i) = H_\numc \cdot i \cdot \opt(c_i)$ and, thus,
	$$
	\alg = \sum_{i=1}^\numc H_\numc \cdot i \cdot \opt(c_i).
	$$
	
	Let~$\chains_{k} = \{c_1,\ldots,c_{k}\}$ for every~$k \in [\numc]$. We first analyze~$\alg(\chains_{3 \cdot \largestinv(\epsilon)})$.
	For the chains in~$\chains_{3 \cdot\largestinv(\epsilon)}$, we get
	\begin{align*}
		\alg(\chains_{3\cdot\largestinv(\epsilon)}) &= \sum_{i = 1}^{3\cdot\largestinv(\epsilon)} H_\numc \cdot i \cdot \opt(c_i)\\
		&\le H_{\numc} \cdot 3 \cdot \largestinv(\epsilon) \cdot \sum_{i = 1}^{3 \cdot\largestinv(\epsilon)} \opt(c_i)\\
		&\le H_{\numc} \cdot 3 \cdot \largestinv(\epsilon) \cdot \opt,
	\end{align*}
	meaning that, for~$\chains_{3 \cdot\largestinv(\epsilon)}$, we achieve the desired competitive ratio. 
	
	Next, consider the chains in~$\chains \setminus \chains_{3 \cdot\largestinv(\epsilon)}$, i.e., the chains~$c_i$ with~$i > 3 \cdot \largestinv(\epsilon)$.
	To analyze the cost for these chains~$C_i$, we continue by lower bounding~$\opt(c_i)$.
	To that end, consider~$\opt(\chains_i)$.
	The definition of~$\largestinv(\epsilon)$ implies that there are at most~$\largestinv(\epsilon)$ chains~$c_j \in \chains_{i}$ with~$W_{c_i} \geq (1 + \epsilon) W_{c_j}$.
	For all other chains~$c_j$ in~$\chains_i$, we have~$\frac{ W_{c_i}}{(1 + \epsilon)} <  W_{c_j}$.
	Thus, there are~$i - \largestinv(\epsilon)$ chains in~$\chains_i$ with a weight of at least~$\frac{ W_{c_i}}{(1 + \epsilon)}$. Since we consider chains with~$i> 3 \cdot \largestinv(\epsilon)$, it holds~$i - \largestinv(\epsilon) \geq 1$.
	We can lower bound~$\opt(\chains_{i})$ by assuming that all such chains consist only of a single job with weight~$\frac{ W_{c_i}}{(1 + \epsilon)}$ and ignoring the up-to~$\largestinv(\epsilon)$ other chains.
	These assumptions only decrease~$\opt(\chains_{i})$. Since in this relaxation all jobs have an equal weight and length, an optimal solution for it processes the jobs in an arbitrary order, giving
	\begin{align*}  
	\opt(\chains_{i}) &\ge \sum_{j=1}^{i-\largestinv(\epsilon)} j \cdot \frac{W_{c_i}}{1+\epsilon}\\
	&= \frac{(i-\largestinv(\epsilon)+1) \cdot (i-\largestinv(\epsilon)) \cdot W_{c_i}}{2 \cdot (1+\epsilon)}\\
	&= \frac{((i+1)\cdot i + \largestinv(\epsilon)^2 - 2 \cdot i \cdot \largestinv(\epsilon)-\largestinv(\epsilon)) \cdot W_{c_i}}{2 \cdot (1+\epsilon)}\\
	&\ge \frac{(i+1)\cdot i  - 3\cdot i \cdot \largestinv(\epsilon)}{2\cdot(1+ \epsilon)} \cdot W_{c_i}.
	\end{align*}

	Since we still assume that each chain has a total processing time of at most~$\sqrt{P}$, we can observe~$\opt(c_i) \le \sqrt{P} \cdot W_{c_i}$. This yields:
	\begin{align*}
		&\frac{2\cdot(1+ \epsilon)}{(i+1)\cdot i  - 3\cdot i \cdot \largestinv(\epsilon)} \cdot \sqrt{P} \cdot \opt(\chains_i)  \ge \opt(c_i) .\\
	\end{align*}
	We can therefore conclude
	\begin{align*}
	\alg(\chains\setminus \chains_{3\cdot\largestinv(\epsilon)}) &= \sum_{i = 3 \cdot \largestinv(\epsilon) + 1}^\numc H_\numc \cdot i \cdot \opt(c_i)\\
	&\le \sum_{i = 3 \cdot \largestinv(\epsilon) + 1}^\numc H_\numc \cdot \frac{2 \cdot (1+\epsilon)}{(i+1)-3\cdot \largestinv(\epsilon)} \cdot \sqrt{P} \cdot \opt(\chains_{i})\\
	&\le 2 \cdot (1+\epsilon) \cdot  H_\numc  \cdot \sqrt{P} \cdot \opt \sum_{i = 3 \cdot \largestinv(\epsilon) + 1}^\numc \frac{1}{(i+1)-3\cdot \largestinv(\epsilon)} \\
	&\le 2 \cdot (1+\epsilon) \cdot  H_\numc \cdot H_{\numc- 3 \largestinv(\epsilon) + 1}  \cdot \sqrt{P} \cdot \opt \\
	&\le 2 \cdot (1+\epsilon) \cdot  H_\numc^2  \cdot \sqrt{P} \cdot \opt . \\
	\end{align*}
	
	We can finish the first part of the proof by combining the bounds for~$\alg(\chains\setminus\chains_{3 \cdot \largestinv(\epsilon)})$ and~$\alg(\chains_{3 \cdot \largestinv(\epsilon)})$:
	\begin{align*}
		\alg(\chains) &= \alg(\chains\setminus\chains_{3 \cdot \largestinv(\epsilon)}) + \alg(\chains_{3 \cdot \largestinv(\epsilon)})\\
		&\le 5 \cdot H_\numc^2 \cdot \sqrt{P} \cdot \max\{1+\epsilon, \largestinv(\epsilon)\} \cdot \opt.
	\end{align*}
	
	\paragraph{Second Part.} It remains to analyze the additional cost incurred by chains with a total processing time of more than~$\sqrt{P}$.
	To that end, consider the set~$J_L$ of jobs that, in any schedule, cannot be started before~$\sqrt{P}$ time units have past. For a job~$j \in J_L$, the predecessors of~$j$ in the chain of~$j$ must have a total processing time of at least~$\sqrt{P}$. 

	Let~$\alg(J_L)$ and~$\opt(J_L)$ denote the weighted completion times of the jobs in~$J_L$ in the optimal solution and the schedule computed by~$\alg$, respectively.
	Then,
	$$
	\frac{\alg(J_L)}{\opt(J_L)} \le \frac{\sum_{j \in J_L} P \cdot w_j}{\sum_{j \in J_L} \sqrt{P} \cdot w_j} = \sqrt{P}.
	$$
	Thus, the additional cost of the jobs in~$J_L$ asymptotically does not worsen the competitive ratio.
\end{proof}

\section{Extension to Parallel Machines}
\label{app:multiple}

We generalize in this section our main results for weight value and adaptive weight order predictions (\Cref{theorem:chain-robin,theorem:out-forests-rr,thm:adaptive-weight-order}) to parallel identical machines. Compared to the single machine, the competitive ratios of these results only increase additively by~$2$. Formally, we consider the scheduling problem where we are given~$m$ parallel identical machines and an algorithm can again assign at any time~$t$ to every front job~$j \in F_t$ a rate~$R_j^t \in [0,1]$. Opposed to the single machine model, the rates now have to satisfy~$\sum_j R_j^t \leq m$ at any time~$t$. Using McNaughton's wrap-around-rule~\cite{mcnaughton1959scheduling}, we can transform these rates to an actual schedule that processes any job at any point in time on at most one machine.

We prove the following generic result, which is a generalization of \Cref{thm:rates-algo}.

\begin{theorem}\label{thm:algo-rates-parallel}
If an algorithm for scheduling weighted jobs with online out-forest precedence constraints on parallel identical machines satisfies at any time~$t$ and for every~$j \in F_j$ that 
\[
	R_j^t < 1 \Longrightarrow m \cdot \frac{w(S(j))}{W(t)} \leq \rho \cdot R_j^t,
\]
where~$R_j^t$ denotes the processing rate of~$j$ at time~$t$, it is at most~$2 + 4\rho$-competitive for minimizing the total weighted completion time.
\end{theorem}

We first apply this theorem to derive competitive ratios of algorithms with weight value and adaptive weight order predictions for parallel identical machines, and finally prove~\Cref{thm:algo-rates-parallel}.

\subsection{Static Weight Values for Chains}

\begin{algorithm}[tb]
	\caption{WDEQ on Chains}\label{alg:chain-robin-identical}
	\begin{algorithmic}[1]
	\REQUIRE Set of chains $\chains$, initial total weight $W_c$ of every chain $c \in \chains$.
	\STATE $t \gets 0$
	\STATE $W_c(t) \gets W_c$ for every chain $c$.
	\WHILE{$U_t \neq \emptyset$}
	\STATE $C \gets \{c_1,\ldots,c_k\}$
	\STATE $m' \gets m$
	\WHILE{$\exists c \in C$ such that $m' \cdot \frac{W_c(t)}{\sum_c W_c(t)} \geq 1$}
	  \STATE Process the first job $j \in I_t$ in chain $c$ with rate $R_j^t = 1$
	  \STATE $C \gets C \setminus \{c\}, m' \gets m' - 1$
	\ENDWHILE
	\STATE For every chain $c \in C$, process the first job $j$ by $R_j^t = m' \cdot \frac{W_c(t)}{\sum_c W_c(t)}$
	\STATE $t \gets t + 1$
	\STATE If some job $j$ in chain $c$ finished, set $W_c(t) \gets W_c(t) - w_j$
	\ENDWHILE
	\STATE Schedule remaining jobs arbitrarily.
	\end{algorithmic}
\end{algorithm}

Recall \Cref{alg:chain-robin} for a single machine, which assigns rates proportional to a chains remaining total weight. These rates fulfill by definition that they can be scheduled on a single machine, i.e., sum up to at most~$1$. On~$m$ parallel identical machines, we intuitively have~$m$ times as much processing power compared to a single machine. A straightforward approach to exploit this observation would be to scale the single machine rates by a factor of~$m$, i.e., process chain~$c$ with rate~$m \cdot \frac{W_c(t)}{\sum_c W_c(t)}$. These rates clearly some up to at most~$m$. However, a single chain might receive a rate more than~$1$, which is disallowed. To circumvent this issue, we cap such rates at~$1$. The final algorithm (\Cref{alg:chain-robin-identical}) therefore first identifies chains with large remaining weight and assigns them rate~$1$, and then splits the remaining processing power among the remaining chains. For non-clairvoyant scheduling of weighted jobs on parallel identical machines without precedence constraints, this approach is known to be~$2$-competitive~\cite{BeaumontBEM12}.

Using the same observation as for \Cref{alg:chain-robin}, we conclude that every job $j$'s rate, when being less than 1, is at most~$m \cdot \frac{w(S(j))}{W(t)}$. Then, \Cref{thm:algo-rates-parallel} implies immediately the following corollary.

\begin{corollary}
	\Cref{alg:chain-robin-identical} is a non-clairvoyant algorithm for the problem of minimizing the total weighted completion time on~$m$ parallel identical machines with online chain precedence constraints with a competitive ratio of at most~$6$, when given access to accurate static weight predictions.
\end{corollary}

\subsection{Adaptive Weight Values for Out-Forests}

In a similar fashion, we can extend the algorithm for adaptive weight values on out-forests to the parallel machine environment, and thereby generalize \Cref{theorem:out-forests-rr}.
To this end, we consider the algorithm which processes every front job $j \in F_t$ with rate $R_j^t = \min\left\{ 1, m \frac{\hW_j}{\sum_{j' \in F_t} \hW_{j'}} \right\}$. This immediately implies that if $R_j^t < 1$, then $R_j^t = m \frac{\hW_j}{\sum_{j' \in F_t} \hW_{j'}} = m \frac{w(S(j))}{W(t)}$ for correct predictions.
We derive the following corollary via \Cref{thm:algo-rates-parallel} with $\rho = 1$. 

\begin{corollary}
	    Given correct weight predictions, there exists a non-clairvoyant $6$-competitive algorithm for minimizing the total weighted completion time on $m$ parallel identical machines with online out-forest precedence constraints .	
\end{corollary}

\subsection{Adaptive Weight Order for Out-Forests}
For adaptive weight order predictions, we essentially use the exact same trick from the previous section and derive \Cref{alg:weight-order-identical}. 

\begin{algorithm}[tb]
	\caption{Adaptive weight order algorithm for parallel identical machines}\label{alg:weight-order-identical}
	\begin{algorithmic}
		\REQUIRE Time~$t$, front jobs $F_t$, adaptive order~$\hpreceq_t$
		\STATE Process every $j \in F_t$ with rate~$R_j^t = \min \left\{1 , \frac{m}{H_{\abs{F_t}} \cdot i_j} \right\}$, where~$i_j$ is the position of $j$ in~$\hpreceq_t$.
	\end{algorithmic}
\end{algorithm}

Observe that using the proof of \Cref{thm:adaptive-weight-order}, we conclude at any time~$t$ and for every~$j \in F_t$ that if~$R_j^t < 1$, then
\[
	m \cdot \frac{w(S(j))}{W(t)} \leq m \cdot H_\numc \cdot \max\{1 + \epsilon, \largestinv(\epsilon)\} \cdot \frac{1}{H_{\abs{F_t}} \cdot i_j} = H_\numc \cdot \max\{1 + \epsilon, \largestinv(\epsilon)\} \cdot R_j^t.
\]
Using \Cref{thm:algo-rates-parallel}, we conclude our result for adaptive weight order predictions.

\begin{corollary}
	For any~$\epsilon > 0$, \Cref{alg:weight-order-identical} has a competitive ratio of at most~$2 + 4 H_\numc \cdot \max\{1 + \epsilon, \largestinv(\epsilon)\}$ for minimizing the total weighted completion time on~$m$ parallel identical machines with online out-forest precedence constraints. \end{corollary}

\subsection{Proof of \Cref{thm:algo-rates-parallel}}

The proof is by dual-fitting and is inspired by \cite{GargGKS19}. Fix an instance, an algorithm which fulfills the stated property, and this algorithm's schedule. We first introduce a linear programming relaxation~\cite{GargGKS19} for our problem on parallel identical machines which are running at lower speed~$\frac{1}{\optspeed}$, for some~$\optspeed \geq 1$. By denoting the value of an optimal solution for the problem with speed~$\frac{1}{\optspeed}$, we conclude~$\opt_\optspeed \leq \optspeed \cdot \opt$.

\begin{alignat}{3}
	\text{min} \quad &\sum_{j,t} w_{j} \cdot t \cdot \frac{x_{j,t}}{p_j} \tag{$\text{P-LP}_\optspeed$}\label{lp-identical} \\ 
	\text{s.t.}  \quad & \sum_{t} \frac{x_{j,t}}{p_{j}} \geq 1   &&\forall j \notag \\
	& \sum_{j} \optspeed \cdot x_{j,t} \leq  m   &&\forall t \notag \\
	& \sum_{s \leq t} \frac{x_{j,s}}{p_j} \geq \sum_{s \leq t} \frac{x_{j',s}}{p_{j'}} && \quad \forall t, (j,j') \in E \notag \\
	& x_{jt} \geq 0 &&\forall j,t \notag 
\end{alignat}

The dual of \eqref{lp-identical} can be written as follows.
  
\begin{alignat}{3}
	\text{max} \quad &\sum_{j} \dualVa_j - m \sum_{t} \dualVb_{t} \tag{$\text{P-DLP}_\optspeed$}\label{dual-identical} \\ 
	\text{s.t.}  \quad &\dualVa_j - w_j \cdot t + \sum_{s \geq t} \left( \sum_{(j,j') \in E} \dualVc_{s,j \to j'} - \sum_{(j', j) \in E} \dualVc_{s,j' \to j} \right) \leq \optspeed \cdot \dualVb_{t} \cdot p_j \qquad &\forall j,t \label{constr:dual-identical} \\ 
	&\dualVa_j, \dualVb_{t}, \dualVc_{t,j \to j'}  \geq 0 \qquad &\forall j,t, (j,j') \in E \notag
\end{alignat}

On a single machine, we use the dual linear program to compare the algorithm' objective to the optimal objective at any time. Unfortunately, showing feasibility of the dual variables which encode the algorithm's schedule crucially require to satisfy~$m \frac{W_{c_j}(t)}{W(t)} \leq \rho \cdot R_j^t$ for some~$\rho \geq 1$ at \emph{any} time and for \emph{every} job~$j$. However, if there is a heavy weight chain, we can only assign it a rate at of at most~$1$. Thus, in general we cannot satisfy this inequality for any constant~$\rho$. To resolve this issue, we instead use a second lower bound for the optimal objective value, and bound the total objective of the algorithm incurred by rates equal to~$1$ against it. This intuitively explains the slightly larger competitive ratio compared to the single machine setting.

We say that a job~$j \in U_t$ \emph{active} at time~$t$  if there exists~$j' \in F_t$ such that~$j \in S(j')$ and~$R_{j'}^{t} < 1$, otherwise it is \emph{inactive}. We remark that $j$ being active at time $t$ implies $t \le C_j$.
For every time~$t$ we denote the set of active jobs by~$A_t$. Note that~$A_t \subseteq U_t$. Intuitively, we will use the dual linear program as lower bound for all jobs in~$A_t$, and the following lower bound every time a job is inactive. For every job~$j$, let~$chain_j$ be the total processing time of a chain's prefix in~$G$ until job~$j$ (including~$j$). Note that~$\opt \geq \sum_j w_j \cdot chain_j$.
 
Let~$\kappa > 1$ be a constant which we fix later. 
We define a dual assignment of \eqref{dual-identical} based on the algorithm's schedule:

\begin{itemize}
	\item~$\dualSa_j = \sum_{s  \ge 0} \dualSa_{j,s}$, where~$\dualSa_{j,s} \begin{cases}
	  w_j \text{ if } j \text{ is active at time } s \\
	  0 \text{ otherwise} 
	\end{cases}$ for every job~$j$,
	\item~$\dualSb_{t} = \frac{1}{m \cdot \kappa} \cdot W(t)$ for every time~$t$, and
	\item 
	$\dualSc_{t, j' \to j} = \begin{cases}
	  0 \quad \text{if } j \notin A_t \text{ or } j' \notin A_t \\
	  W(S(j)) 
	\end{cases}$ for every time~$t$ and edge~$(j', j) \in E$.
  \end{itemize}
  
\begin{lemma}\label{lemma:chain-robin-id-dual-objective}
	$\sum_{j} \dualSa_j - m \sum_{t} \dualSb_{t} + \sum_j w_j \cdot chain_j \geq (1 - \frac{1}{\kappa}) \cdot \alg$.
\end{lemma}
  
\begin{proof}
  Since the weight~$w_j$ of a job~$j$ is contained in~$W(t)$ if~$t \leq C_j$, we conclude~$\sum_t \dualSb_{t} = \frac{1}{m \cdot \kappa}\alg$. Now consider a job~$j$ and a time~$t \leq C_j$. If~$j$ is active at time~$t$, by definition~$\dualSa_{j,s} = w_j$. Otherwise,~$j$ or some predecessor of~$j$ is being processed with rate~$1$. Hence, we conclude that the number of times before~$C_j$ when~$j$ is inactive is at most~$chain_j$. Therefore~$w_j C_j \leq \dualSa_j +  w_j \cdot chain_j$ and~$\alg \leq \sum_j \dualSa_j + \sum_j w_j \cdot chain_j$.
\end{proof}
  
  \begin{lemma}\label{lemma:chain-robin-id-dual-feasible}
	Assigning~$\dualVa_j = \dualSa_j$,~$\dualVb_{t} = \dualSb_{t}$ and~$\dualVc_{t,j \to j'} = \dualSc_{t,j \to j'}$ gives a feasible solution for~\eqref{dual-identical} if there exists a~$\rho \geq 1$ such that~$\alpha \geq \kappa \cdot \rho$ and the algorithm's rates satisfy at any time~$t$ for every~$j \in F_t$ 
	\begin{equation}
			R_j^t < 1 \Longrightarrow m \cdot \frac{W_{c_j}(t)}{W(t)} \leq \rho \cdot R_j^t. \label{eq:rate-condition-identical}
	\end{equation}
  \end{lemma}
  
\begin{proof}
Since our defined variables are non-negative by definition, it suffices to show that this assignment satisfies~\eqref{constr:dual-identical}.
Fix a job~$j$ and a time~$t \geq 0$. Let~$c$ be~$j$'s chain.
By observing that~$\dualSa_j - t \cdot w_j \leq \sum_{s \geq t} \dualSa_{j,s}$, it suffices to verify
\begin{equation}
  \sum_{s \geq t} \left(\dualSa_{j,s} +  \sum_{(j, j') \in E} \dualSc_{s,j \to j'} - \sum_{(j', j) \in E} \dualSc_{s,j' \to j}  \right) \leq \optspeed \cdot \dualSb_{t} \cdot p_j. \label{eq:dual-red1-identical}
 \end{equation}
  
To this end, we consider the terms of the left side for all times~$s \geq t$ separately. For any~$s$ with~$s > C_j$, the left side of \eqref{eq:dual-red1-identical} is zero, because~$\dualSa_{j,s} = 0$ and~$j \notin A_s$. 
Also note that at any time~$s \in [0,C_j]$ when~$j$ is inactive, the left side of \eqref{eq:dual-red1-identical} is again zero. Therefore, we restrict our attention in the following to times where $j$ is active.

Let~$t_j^*$ be the first point in time after~$t$ when~$j$ is available, and let~$s \in [0, t_j^*)$ such that~$j$ is active at time~$s$. Observe that~$j \in A_s$ and since each vertex in an out-forest has at most one direct predecessor, there must be a unique job~$j_1 \in A_s$ with~$(j_1, j) \in E$. Thus~$\dualSc_{s,j_1 \to j} = w(S(j)) =  \sum_{(j', j) \in E} \dualSc_{s,j' \to j}$ and $\dualSc_{s,j \to j'} = w(S(j'))$ for every $j'$ with $(j,j') \in E$, because $j'$ must also be unfinished and active, hence $j' \in A_s$.
Using the fact that the sets $S(j')$ are pairwise disjoint for all direct successors $j'$ of $j$, we conclude $ \sum_{(j, j') \in E} \dualSc_{s,j' \to j} = w(S(j))-w_j$ and, thus,
\[
  \dualSa_{j,s} + \sum_{(j, j') \in E} \dualSc_{s,j \to j'} - \sum_{(j', j) \in E} \dualSc_{s,j' \to j} \leq w_j - w_j = 0. 
\]
  
By defining~$T_j = \{ t^*_j \leq s \leq C_j \mid j \text{ active at time } s \}$, we conclude that proving \eqref{eq:dual-red1-identical} reduces to proving
  \begin{equation}
	  \sum_{s \in T_j} \left(w_j + \sum_{(j, j') \in E} \dualSc_{s,j \to j'} - \sum_{(j', j) \in E} \dualSc_{s,j' \to j} \right) \leq \optspeed \cdot \dualSb_{t} \cdot p_j. \label{eq:dual-red2-identical}
  \end{equation}
  
  Now, let~$s \in T_j$. There cannot be an unfinished job preceding~$j$, thus~$\sum_{(j', j) \in E} \dualSc_{s,j' \to j} = 0$. Observe that if there is a job~$j' \in A_s$ with~$(j, j') \in E$, the fact that~$j \in A_s$ gives~$\dualSc_{s,j \to j'} = w(S(j'))$. By again exploiting the out-forest topology, this yields
  \[ 
	  w_j + \sum_{(j, j') \in E} \dualSc_{s,j \to j'} - \sum_{(j', j) \in E} \dualSc_{s,j' \to j} = w_j + \sum_{(j, j') \in E} w(S(j')) = w(S(j)),
  \]
  and hence implies that the left side of~\eqref{eq:dual-red2-identical} is equal to~$\sum_{s \in T_j} w(S(j))$.
  
  The facts that~$W(t_1) \geq W(t_2)$ at any~$t_1 \leq t_2$ and that~$j$ is processed by~$R_{j}^{s} < 1$ units at any time~$s \in T_j$ imply combined with \eqref{eq:rate-condition-identical} the following inequality:
  \begin{equation*}
	  m \cdot \sum_{s \in T_j} \frac{w(S(j))}{W(t)} \leq m \cdot \sum_{s \in T_j} \frac{w(S(j))}{W(s)} \leq \sum_{s \in T_j} \rho \cdot R_j^s \leq \rho \cdot p_j. 
  \end{equation*}
  Rearranging it, using the definition of~$\dualSb_t$ and the bound on~$\optspeed$ gives
	\[
		\sum_{s \in T_j} w(S(j)) \leq \rho \cdot p_j \cdot \frac{1}{m} \cdot W(t) = \rho \cdot \kappa \cdot p_j \cdot  \dualSb_t \leq \optspeed \cdot p_j \cdot \dualSb_t,
	\]
	which implies \eqref{eq:dual-red2-identical} and thus proves the statement.
  \end{proof}
  
\begin{proof}[Proof of \Cref{thm:algo-rates-parallel}]
	We set~$\alpha = \kappa \cdot \rho$.
	The facts that~$\optspeed \cdot \opt \geq \opt_\optspeed$ and~$\sum_j w_j \cdot chain_j \leq \opt$, weak duality and \Cref{lemma:chain-robin-id-dual-feasible} imply
	\[
		(1 + \alpha) \cdot \opt \geq \opt_\alpha + \opt 
		\geq \sum_{j} \dualSa_j - m \sum_{i,t} \dualSb_{t} + \sum_j w_j \cdot chain_j.
	\]
	Using \Cref{lemma:chain-robin-id-dual-objective} concludes
	\[
		(1 + \alpha) \cdot \opt \geq \left(1 - \frac{1}{\kappa} \right) \cdot \alg,
	\]
	which, choosing~$\kappa = 2$, can be rearranged to~$\alg \leq \left( \frac{1 + 2 \rho}{1 - \frac{1}{2}} \right) \cdot \opt = \left( 2 + 4 \rho \right) \cdot \opt$.
\end{proof}

\section{Time-Sharing without Monotonicity}\label{app:monotonicity}

We give a proof of the following theorem. 
The proofs given in~\cite{PurohitSK18,LindermayrM22} require both algorithms to be monotone, that is, for every two instances on the same set of jobs but with processing requirements $\{p'_j\}_j$ and $\{p_j\}_j$ such that for all jobs holds $p'_j \leq p_j$, the algorithm's objective value for $\{p'_j\}_j$ must be less or equal than the algorithm's objective value for $\{p_j\}_j$.
In contrast, we do \emph{not} require monotonicity. This even holds when scheduling jobs on unrelated machines, where jobs receive a different progress on different machines. The only requirement on both algorithms is that they are \emph{progress-aware}. That is, they can track the exact progress jobs have received so far at any point in time, which is true for most algorithms in almost every model.

\begin{theorem}
	Given two deterministic progress-aware algorithms $\A$ and $\B$ with competitive ratios~$\rho_\A$ and~$\rho_\B$ for minimizing the total weighted completion time with online precedence constraints, there exists an algorithm for the same problem with a competitive ratio of at most~$2 \cdot \min\{\rho_\A, \rho_\B\}$.
\end{theorem}

\begin{proof}
We consider the algorithm which simulates each of both algorithms \A and \B with a rate equal to~$\frac{1}{2}$ on the same set of jobs, i.e., in every time step it executes algorithm~$\A$ in the first half of the timestep, and algorithm~$\B$ in the second half. Additionally, it keeps track of how much each algorithm advances every job.
For any job $j$, both algorithms ignore the progress in the processing of $j$ made by the other algorithm. In particular, if $\A$ finishes job $j$, then $\B$ still simulates the processing of~$j$ until the total  time spend by $\B$ on the (partially simulated) processing of $j$ equals $p_j$ (and vice versa if $\B$ finishes a job before $\A$). We assume that an algorithm only can start successors of a job $j$ once the (partially simulated) processing time spend \emph{by the algorithm} on $j$ is at least $p_j$. We remark that his requires that the main algorithm is able to manipulate the input for the sub-algorithms $\A$ and $\B$.
This ensures that the simulated completion times~$\tilde{C}^\A_j$,~$\tilde{C}^\B_j$ for both algorithms of every job~$j$ are exactly doubled compared to the completion times~$C^\A_j$,~$C^\B_j$ of~$j$ in independent schedules of both algorithms, i.e.,~$\tilde{C}^\A_j = 2\cdot C^\A_j$ and~$\tilde{C}^\B_j = 2\cdot C^\B_j$. In the actual combined schedule, job~$j$ clearly completed not later than~$\min\{\tilde{C}^\A_j, \tilde{C}^\B_j\} = 2 \cdot \min\{C^\A_j, C^\B_j\}$, implying the stated bound on its competitive ratio.
\end{proof}

We finally note that this argument also works when the share each of both algorithms receives is parameterized by some~$\lambda \in [0,1]$.

\end{document}